\documentclass[12pt]{amsart}
\usepackage{}
\usepackage{tikz}
  \usetikzlibrary{matrix,arrows,positioning,shapes.geometric}
\usepackage{amsmath}
\usepackage{amsfonts}
\usepackage{amssymb}
\usepackage[all]{xy}           

\usepackage{bbding}
\usepackage{txfonts}
\usepackage{amscd}

\usepackage[shortlabels]{enumitem}
\usepackage{ifpdf}
\ifpdf
  \usepackage[colorlinks,final,backref=page,hyperindex]{hyperref}
\else
  \usepackage[colorlinks,final,backref=page,hyperindex,hypertex]{hyperref}
\fi
\usepackage{tikz}
\usepackage[active]{srcltx}

\makeatletter

\newtheorem{thm}{Theorem}[section]
\newtheorem{lem}[thm]{Lemma}

\newtheorem{pro}[thm]{Proposition}

\newtheorem{rmk}[thm]{Remark}
\newtheorem{defi}[thm]{Definition}

\newcommand {\emptycomment}[1]{}

\newcommand{\be }{\begin{equation}}
\newcommand{\ee }{\end{equation}}

\newcommand{\g}{\mathfrak g}

\newcommand{\huaC}{{\mathcal{C}}}

\newcommand{\Id}{\rm{Id}}

\newcommand{\br}[1]{   [ \cdot,    \cdot  ]   }

\newcommand{\Hom}{\mathrm{Hom}}

\newcommand{\gl}{\mathfrak {gl}}

\newcommand{\reg}{\mathrm{reg}}

\begin{document}

\title[]{The full cohomology, abelian extensions and formal deformations of Hom-pre-Lie algebras}

\author{Shanshan Liu}
\address{College of Mathematics and Systems Science, Shandong University of Science and Technology, Qingdao 266590, Shandong, China}
\email{shanshanmath@163.com}

\author{Abdenacer Makhlouf}
\address{University of Haute Alsace, IRIMAS-Mathematics department, Mulhouse, France}
\email{abdenacer.makhlouf@uha.fr}

\author{Lina Song}
\address{Department of Mathematics, Jilin University, Changchun 130012, Jilin, China}
\email{songln@jlu.edu.cn}


\begin{abstract}
The main purpose of this paper is to provide a full  cohomology of a Hom-pre-Lie algebra with coefficients in a given  representation. This new type of cohomology exploit strongly the Hom-type structure and fits perfectly with simultaneous deformations of the multiplication and the homomorphism defining a Hom-pre-Lie algebra.  Moreover, we show that its second cohomology group classifies abelian extensions of a Hom-pre-Lie algebra by a  representation.

\end{abstract}


\maketitle
\vspace{-5mm}
\tableofcontents

\allowdisplaybreaks


\section{Introduction}

The notion  of a pre-Lie algebra has been introduced  independently  by M. Gerstenhaber in  deformation theory  of rings and algebras \cite{Gerstenhaber63} and by Vinberg, under the name of left-symmetric algebra,  in his theory of homogeneous convex  cones \cite{Vinberg}. Its defining identity is weaker than associativity and lead also to a Lie algebra using commutators. This algebraic structure describes some properties of cochains space in Hochschild cohomology of an associative algebra, rooted trees and vector fields on affine spaces.
 Moreover, it is playing an increasing role in algebra, geometry and physics due to their applications in nonassociative algebras, combinatorics,  numerical Analysis and quantum field theory, see also \cite{Pre-lie algebra in geometry,Bai,Bai2,CK}.

 Hom-type algebras appeared naturally when studying $q$-deformations of some algebras of vector fields, like Witt and Virasoro algebras. It turns out that the Jacobi identity is no longer satisfied, these new structures involving a bracket and a linear map satisfy a twisted version of the Jacobi identity and define a so called Hom-Lie algebras which form a wider class, see \cite{HLS,LD1}.   Hom-pre-Lie algebras were introduced in \cite{MS2} as a class of Hom-Lie admissible algebras, and play important roles in the study of Hom-Lie bialgebras and Hom-Lie 2-algebras \cite{Cai-Sheng,sheng1,SC}. Recently Hom-pre-Lie algebras were studied from several aspects. Cohomologies of Hom-pre-Lie algebras were studied in \cite{LSS}; The geometrization of Hom-pre-Lie algebras was studied in \cite{Qing}; Universal $\alpha$-central extensions of Hom-pre-Lie algebras were studied in \cite{sunbing}; Hom-pre-Lie bialgebras were studied in \cite{LSSbi,QH}. Furthermore, connections between (Hom-)pre-Lie algebras and various algebraic structures have been established and discussed; like with Rota-Baxter operators, $\mathcal{O}$-operators, (Hom-)dendriform algebras, (Hom-)associative algebras and Yang-Baxter equation.

The cohomology of pre-Lie algebras was defined in \cite{DA}  and generalized  in a straightforward way to Hom-pre-Lie algebras in \cite{LSS}. Note that the cohomolgoy given there has some restrictions: the second cohomology group can only control deformations of the multiplication, and it can not be applied to study simultaneous deformations of both the multiplication and the homomorphism in a Hom-pre-Lie algebra. The main purpose of this paper is to define a new type of cohomology of Hom-pre-Lie algebras which is richer than previous one. The obtained cohomology is called the full cohomology and allows to control simultaneous deformations of both the multiplication and the homomorphism in a Hom-pre-Lie algebra. 
This type of cohomology  was established  for Hom-associative algebras and Hom-Lie algebras in \cite{BM,Bene}, and called respectively  $\alpha$-type Hochschild cohomology and $\alpha$-type Chevalley-Eilenberg  cohomology. See \cite{AEM,MS1,SLN} for more studies on deformations and extensions of Hom-Lie algebras.

The paper is organized as follows. In Section \ref{sec:full}, first we recall some basics about Hom-Lie algebras, Hom-pre-Lie algebras and representations, and then provide our main result defining the full cohomology of  a Hom-pre-Lie algebra with coefficients in a given  representation. In Section \ref{sec:def}, we study one parameter  formal deformations of a Hom-pre-Lie algebra, where both the defining multiplication and  homomorphism are deformed, using formal power series. We show that the full cohomology of Hom-pre-Lie algebras controls these simultaneous deformations. Moreover, a relationship between deformations of a Hom-pre-Lie algebra and deformations of  its sub-adjacent Lie algebra is established. Section \ref{sec:ext} deals with abelian extensions of Hom-pre-Lie algebras. We show that the full cohomology fits perfectly and its second cohomology group classifies abelian extensions of a Hom-pre-Lie algebra by a given representation. The proof of the  key Lemma to show that the four operators define a cochain complex  is lengthy and it is  given in the  Appendix.

\vspace{2mm}
\noindent
{\bf Acknowledgements. }  We give warmest thanks to Yunhe Sheng   for helpful comments that improve the paper.   This work is supported by
 National Natural Science Foundation of China (Grant No. 12001226).

\section{The full cohomology of Hom-pre-Lie algebras}\label{sec:full}

In this section, first we recall some basic facts about Hom-Lie algebras and Hom-pre-Lie algebras. Then we introduce the full cohomology of Hom-pre-Lie algebras, which will be used to classify infinitesimal deformations and abelian extensions of Hom-pre-Lie algebras.

\begin{defi}{\rm(\cite{HLS})}
A {\bf Hom-Lie algebra} is a triple $(\g,[\cdot,\cdot]_{\g},\phi_{\g})$ consisting of a vector space $\g$, a skew-symmetric bilinear map $[\cdot,\cdot]_{\g}:\wedge^2\g\longrightarrow \g$ and an algebra homomorphism $\phi_{\g}:\g\longrightarrow \g $, satisfying:
  \begin{equation}
[\phi_{\g}(x),[y,z]_{\g}]_{\g}+[\phi_{\g}(y),[z,x]_{\g}]_{\g}+[\phi_{\g}(z),[x,y]_{\g}]_{\g}=0,\quad
\forall~x,y,z\in \g.
\end{equation}
A Hom-Lie algebra $(\g,[\cdot,\cdot]_{\g},\phi_{\g})$ is said to be {\bf regular} if $\phi_{\g}$ is invertible.
  \end{defi}
 \begin{defi}{\rm(\cite{sheng3})}\label{defi:hom-lie representation}
 A {\bf representation} of a Hom-Lie algebra $(\g,[\cdot,\cdot]_{\g},\phi_{\g})$ on
 a vector space $V$ with respect to $\beta\in\gl(V)$ is a linear map
  $\rho:\g\longrightarrow \gl(V)$, such that for all
  $x,y\in \g$, the following equalities are satisfied:
\begin{eqnarray}
\label{hom-lie-rep-1}\rho(\phi_{\g}(x))\circ \beta&=&\beta\circ \rho(x),\\
\label{hom-lie-rep-2}\rho([x,y]_{\g})\circ \beta&=&\rho(\phi_{\g}(x))\circ\rho(y)-\rho(\phi_{\g}(y))\circ\rho(x).
\end{eqnarray}
  \end{defi}
We denote a representation of a Hom-Lie algebra $(\g,[\cdot,\cdot]_{\g},\phi_{\g})$ by a triple  $(V,\beta,\rho)$.
\begin{defi}{\rm(\cite{MS2})}
A {\bf Hom-pre-Lie algebra} $(A,\cdot,\alpha)$ is a vector space $A$ equipped with a bilinear product $\cdot:A\otimes A\longrightarrow A$, and $\alpha\in \gl(A)$, such that for all $x,y,z\in A$, $\alpha(x \cdot y)=\alpha(x)\cdot \alpha(y)$ and the following equality is satisfied:
\begin{eqnarray}
(x\cdot y)\cdot \alpha(z)-\alpha(x)\cdot (y\cdot z)=(y\cdot x)\cdot \alpha(z)-\alpha(y)\cdot (x\cdot z).
\end{eqnarray}
A Hom-pre-Lie algebra $(A,\cdot,\alpha)$ is said to be {\bf regular} if $\alpha$ is invertible.
\end{defi}
Let $(A,\cdot,\alpha)$ be a Hom-pre-Lie algebra. The commutator $[x,y]_C=x\cdot y-y\cdot x$ defines a Hom-Lie algebra $(A,[\cdot,\cdot]_C,\alpha)$, which is denoted by $A^C$ and called  the {\bf sub-adjacent Hom-Lie algebra} of $(A,\cdot,\alpha)$.
\begin{defi}{\rm(\cite{LSS})}
 A morphism from a Hom-pre-Lie algebra $(A,\cdot,\alpha)$ to a Hom-pre-Lie algebra $(A',\cdot',\alpha')$ is a linear map $f:A\longrightarrow A'$ such that for all
 $x,y\in A$, the following equalities are satisfied:
\begin{eqnarray}
\label{homo-1}f(x\cdot y)&=&f(x)\cdot' f(y),\hspace{3mm}\forall x,y\in A,\\
\label{homo-2}f\circ \alpha&=&\alpha'\circ f.
\end{eqnarray}
\end{defi}
\begin{defi}{\rm(\cite{QH})}\label{defi:hom-pre representation}
 A {\bf representation} of a Hom-pre-Lie algebra $(A,\cdot,\alpha)$ on a vector space $V$ with respect to $\beta\in\gl(V)$ consists of a pair $(\rho,\mu)$, where $\rho:A\longrightarrow \gl(V)$ is a representation of the sub-adjacent Hom-Lie algebra $A^C$ on $V$ with respect to $\beta\in\gl(V)$, and $\mu:A\longrightarrow \gl(V)$ is a linear map, satisfying, for all $x,y\in A$:
\begin{eqnarray}
\label{rep-1}\beta\circ \mu(x)&=&\mu(\alpha(x))\circ \beta,\\
\label{rep-2}\mu(\alpha(y))\circ\mu(x)-\mu(x\cdot y)\circ \beta&=&\mu(\alpha(y))\circ\rho(x)-\rho(\alpha(x))\circ\mu(y).
\end{eqnarray}
\end{defi}
We denote a representation of a Hom-pre-Lie algebra $(A,\cdot,\alpha)$ by a quadruple  $(V,\beta,\rho,\mu)$. Furthermore, Let $L,R:A\longrightarrow \gl(A)$ be linear maps, where $L_xy=x\cdot y, R_xy=y\cdot x$. Then $(A,\alpha,L,R)$ is also a representation, which we call the regular representation.

Let $(A,\cdot,\alpha)$ be a Hom-pre-Lie algebra. In the sequel, we will also denote the Hom-pre-Lie algebra multiplication $\cdot$ by $\omega$.

Let $(V,\beta,\rho,\mu)$ be a representation of a Hom-pre-Lie algebra $(A,\omega,\alpha)$. We define $\huaC^n_\omega(A;V)$ and $\huaC^n_\alpha(A;V)$ respectively by
$$\huaC^n_\omega(A;V)=\Hom(\wedge^{n-1} A\otimes A,V),\quad \huaC^n_\alpha(A;V)=\Hom(\wedge^{n-2} A\otimes A,V),\quad
 \forall n\geq 2.$$
Define the set of cochains $\tilde{\huaC}^n(A;V)$ by
\begin{equation}\left\{\begin{array}{ll}
 & \tilde{\huaC}^n(A;V)=\huaC^n_\omega(A;V)\oplus \huaC^n_\alpha(A;V),\quad
 \forall n\geq 2,\\
& \tilde{\huaC}^1(A;V)=\Hom(A,V).
\end{array}\right.
\end{equation}
For all $(\varphi,\psi) \in \tilde{\huaC}^n(A;V),x_1,\dots,x_{n+1}\in A$, we define $\partial_{\omega\omega}:\huaC^n_\omega(A;V)\longrightarrow \huaC^{n+1}_\omega(A;V)$ by
\begin{eqnarray*}
&&(\partial_{\omega\omega}\varphi)(x_1,\dots,x_{n+1})\\
&&=\sum_{i=1}^n(-1)^{i+1}\rho(\alpha^{n-1}(x_i))\varphi(x_1,\dots,\widehat{x_i},\dots,x_{n+1})\\
&&+\sum_{i=1}^n(-1)^{i+1}\mu(\alpha^{n-1}(x_{n+1}))\varphi(x_1,\dots,\widehat{x_i},\dots,x_n,x_i)\\
&&-\sum_{i=1}^n(-1)^{i+1} \varphi(\alpha(x_1),\dots,\widehat{\alpha(x_i)}\dots,\alpha(x_n),x_i\cdot x_{n+1})\\
&&+\sum_{1\leq i<j\leq n}(-1)^{i+j} \varphi([x_i,x_j]_C,\alpha(x_1),\dots,\widehat{\alpha(x_i)},\dots,\widehat{\alpha(x_j)},\dots,\alpha(x_{n+1})),
\end{eqnarray*}
define $\partial_{\alpha\alpha}:\huaC^n_\alpha(A;V)\longrightarrow \huaC^{n+1}_\alpha(A;V)$ by
\begin{eqnarray*}
&&(\partial_{\alpha\alpha} \psi)(x_1,\dots,x_n)\\
&&=\sum_{i=1}^{n-1}(-1)^i\rho(\alpha^{n-1}(x_i))\psi(x_1,\dots,\widehat{x_i},\dots,x_n)\\
&&+\sum_{i=1}^{n-1}(-1)^i\mu(\alpha^{n-1}(x_n))\psi(x_1,\dots,\widehat{x_i},\dots,x_{n-1},x_i)\\
&&-\sum_{i=1}^{n-1}(-1)^i \psi(\alpha(x_1),\dots,\widehat{\alpha(x_i)},\dots,\alpha(x_{n-1}),x_i\cdot x_n)\\
&&+\sum_{1\leq i<j\leq n-1}(-1)^{i+j-1} \psi([x_i,x_j]_C,\alpha(x_1),\dots,\widehat{\alpha(x_i)},\dots,\widehat{\alpha(x_j)},\dots,\alpha(x_n)),
\end{eqnarray*}
define $\partial_{\omega\alpha}:\huaC^n_\omega(A;V)\longrightarrow \huaC^{n+1}_\alpha(A;V)$ by
\begin{eqnarray*}
(\partial_{\omega\alpha} \varphi)(x_1,\dots,x_n)=\beta\varphi(x_1,\dots,x_n)-\varphi(\alpha(x_1),\dots,\alpha(x_n)),
\end{eqnarray*}
and define $\partial_{\alpha\omega}:\huaC^n_\alpha(A;V)\longrightarrow \huaC^{n+1}_\omega(A;V)$ by
\begin{eqnarray*}
&&(\partial_{\alpha\omega} \psi)(x_1,\dots,x_{n+1})\\
&&=\sum_{1\leq i<j\leq n}(-1)^{i+j}\rho([\alpha^{n-2}(x_i),\alpha^{n-2}(x_j)]_C) \psi(x_1,\dots,\hat{x_i},\dots,\hat{x_j},\dots,x_{n+1})\\
&&+\sum_{1\leq i<j\leq n}(-1)^{i+j}\mu(\alpha^{n-2}(x_i)\cdot\alpha^{n-2}(x_{n+1}))\psi(x_1,\dots,\hat{x_i},\dots,\hat{x_j},\dots,x_n,x_j)\\
&&-\sum_{1\leq i<j\leq n}(-1)^{i+j}\mu(\alpha^{n-2}(x_j)\cdot\alpha^{n-2}(x_{n+1}))\psi(x_1,\dots,\hat{x_i},\dots,\hat{x_j},\dots,x_n,x_i).
\end{eqnarray*}
Define the operator $\tilde{\partial}:\tilde{\huaC}^n(A;V)\longrightarrow  \tilde{\huaC}^{n+1}(A;V)$ by
\begin{eqnarray}
\label{complex-cohomology-1} \tilde{\partial}(\varphi,\psi)&=&(\partial_{\omega\omega}\varphi+\partial_{\alpha\omega}\psi,\partial_{\omega\alpha}\varphi+\partial_{\alpha\alpha}\psi),\quad
 \forall \varphi \in \huaC^n_\omega(A;V), \psi \in \huaC^n_\alpha(A;V), n\geq 2,\\
\label{complex-cohomology-2}\tilde{\partial}(\varphi)&=&(\partial_{\omega\omega}\varphi,\partial_{\omega\alpha}\varphi),\quad
 \forall \varphi \in \Hom(A,V).
\end{eqnarray}
The following diagram will explain the above operators:
\begin{center}
   \begin{tikzpicture}
      $$\matrix (m) [matrix of math nodes,row sep=3em,column sep=5em,minimum width=2em] {
			 \huaC^n_\omega(A;V) &  \huaC^{n+1}_\omega(A;V) & \huaC^{n+2}_\omega(A;V)  \\
		\huaC^n_\alpha(A;V) &  \huaC^{n+1}_\alpha(A;V) & \huaC^{n+2}_\alpha(A;V)\\
		};
		\path[->,auto] (m-2-1) edge node[swap]{$\partial_{\alpha\alpha}$}                  (m-2-2);
        \path[->,auto] (m-2-2) edge node[swap]{$\partial_{\alpha\alpha}$}                  (m-2-3);
		\path[->,above] (m-2-1) edge node[below] {$\partial_{\alpha\omega}$}         (m-1-2);
        \path[->,above] (m-2-2) edge node[below] {$\partial_{\alpha\omega}$}         (m-1-3);
		\path[->,auto] (m-1-1) edge node {$\partial_{\omega\omega}$}                  (m-1-2);
        \path[->,auto] (m-1-2) edge node {$\partial_{\omega\omega}$}                  (m-1-3);
        \path[->,below] (m-1-1) edge node[above] {$\partial_{\omega\alpha}$}                  (m-2-2);
        \path[->,below] (m-1-2) edge node[above] {$\partial_{\omega\alpha}$}                  (m-2-3);
		\path   (m-2-1) edge[draw=none]   node [sloped] {$\oplus$} (m-1-1);
		\path   (m-2-2) edge[draw=none]   node [sloped] {$\oplus$} (m-1-2);
        \path   (m-2-3) edge[draw=none]   node [sloped] {$\oplus$} (m-1-3);$$
	\end{tikzpicture}
\end{center}

\begin{lem}\label{lem:cohomology-computation}
With the above notations, we have
\begin{eqnarray}
\partial_{\omega\omega}\circ \partial_{\omega\omega}+\partial_{\alpha\omega}\circ \partial_{\omega\alpha}&=&0,\\
\partial_{\omega\omega}\circ \partial_{\alpha\omega}+\partial_{\alpha\omega}\circ \partial_{\alpha\alpha}&=&0,\\
\partial_{\omega\alpha}\circ \partial_{\omega\omega}+\partial_{\alpha\alpha}\circ \partial_{\omega\alpha}&=&0,\\
\partial_{\omega\alpha}\circ \partial_{\alpha\omega}+\partial_{\alpha\alpha}\circ \partial_{\alpha\alpha}&=&0.
\end{eqnarray}
\end{lem}
\begin{proof}
The proof is given in Appendix.
\end{proof}
\begin{thm}\label{thm:operator}
The operator $\tilde{\partial}:\tilde{C}^n(A;V)\longrightarrow \tilde{C}^{n+1}(A;V)$ defined as above satisfies $\tilde{\partial}\circ\tilde{\partial}=0$.
\end{thm}
\begin{proof}
When $n\geq 2$, for all $(\varphi,\psi)\in \tilde{\huaC}^n(A;V)$, by \eqref{complex-cohomology-1} and Lemma \ref{lem:cohomology-computation}, we have
\begin{eqnarray*}
 \tilde{\partial}\circ\tilde{\partial}(\varphi,\psi) &=&\tilde{\partial}(\partial_{\omega\omega}\varphi+\partial_{\alpha\omega}\psi,\partial_{\omega\alpha}\varphi+\partial_{\alpha\alpha}\psi)\\
  &=&(\partial_{\omega\omega}\partial_{\omega\omega}\varphi+\partial_{\omega\omega}\partial_{\alpha\omega}\psi+\partial_{\alpha\omega}\partial_{\omega\alpha}\varphi+\partial_{\alpha\omega}\partial_{\alpha\alpha}\psi,\\
  &&\partial_{\omega\alpha}\partial_{\omega\omega}\varphi+\partial_{\omega\alpha}\partial_{\alpha\omega}\psi+\partial_{\alpha\alpha}\partial_{\omega\alpha}\varphi+\partial_{\alpha\alpha}\partial_{\alpha\alpha}\psi)\\
  &=&0.
\end{eqnarray*}
When $n=1$, for all $\varphi \in \Hom(A,V)$,  by \eqref{complex-cohomology-2} and Lemma \ref{lem:cohomology-computation}, we have
\begin{eqnarray*}
 \tilde{\partial}\circ\tilde{\partial}(\varphi) &=&\tilde{\partial}(\partial_{\omega\omega}\varphi,\partial_{\omega\alpha}\varphi)\\
  &=&(\partial_{\omega\omega}\partial_{\omega\omega}\varphi+\partial_{\alpha\omega}\partial_{\omega\alpha}\varphi,
  \partial_{\omega\alpha}\partial_{\omega\omega}\varphi+\partial_{\alpha\alpha}\partial_{\omega\alpha}\varphi)\\
  &=&0.
\end{eqnarray*}
This finishes the proof.
\end{proof}
We denote the set of closed $n$-cochains by $\tilde{Z}^n(A;V)$ and the set of exact $n$-cochains by $\tilde{B}^n(A;V)$. We denote by $\tilde{H}^n(A;V)=\tilde{Z}^n(A;V)/\tilde{B}^n(A;V)$ the corresponding cohomology groups.

 \begin{defi}
Let $(V,\beta,\rho,\mu)$ be a representation of a Hom-pre-Lie algebra $(A,\cdot,\alpha)$. The cohomology of the cochain complex $(\oplus_{n=1}^\infty\tilde{C}^n(A;V),\partial)$ is called the {\bf full cohomology} of the Hom-pre-Lie algebra $(A,\cdot,\alpha)$ with coefficients in the representation $(V,\beta,\rho,\mu)$.
 \end{defi}
We use $\tilde{\partial}_{\reg}$ to denote the coboundary operator of the Hom-pre-Lie algebra $(A,\cdot,\alpha)$ with  coefficients in the regular representation. The corresponding cohomology group will be denoted by $\tilde{H}^n(A;A)$.

\begin{rmk}
  Compared with the cohomology theory of Hom-pre-Lie algebras studied in \cite{LSS}, the above full cohomology contains more information. In the next section, we will see that the second cohomology group can control simultaneous deformations of the multiplication and the homomorphism in a Hom-pre-Lie algebra.
\end{rmk}

\section{Formal deformations of Hom-pre-Lie algebras}\label{sec:def}
In this section, we study formal deformations of Hom-pre-Lie algebras using the cohomology theory established in the last section. We show that the infinitesimal of a formal deformation is a 2-cocycle and depends only on its cohomology class. Moreover, if the cohomology group $\tilde{H}^2(A;A)$ is trivial, then the Hom-pre-Lie algebra is rigid.

\begin{defi}
Let $(A,\omega,\alpha)$ be a Hom-pre-Lie algebra, $\omega_t=\omega+\sum_{i=1}^{+\infty} \omega_it^i:A[[t]]\otimes A[[t]]\longrightarrow A[[t]]$ be a $\mathbb K[[t]]$-bilinear map and $\alpha_t=\alpha+\sum_{i=1}^{+\infty} \alpha_it^i:A[[t]]\longrightarrow A[[t]]$ be a $\mathbb K[[t]]$-linear map, where $\omega_i:A\otimes A\longrightarrow A$ and $\alpha_i:A\longrightarrow A$ are linear maps. If $(A[[t]],\omega_t,\alpha_t)$ is still a Hom-pre-Lie algebra, we say that $\{\omega_i,\alpha_i\}_{i\geq1}$ generates a {\bf $1$-parameter formal deformation} of a Hom-pre-Lie algebra $(A,\omega,\alpha)$.
\end{defi}

If $\{\omega_i,\alpha_i\}_{i\geq1}$ generates a $1$-parameter formal deformation of a Hom-pre-Lie algebra $(A,\omega,\alpha)$, for all $x,y,z\in A$ and $n=1,2,\dots$,  we have
\begin{equation}\label{cocycle-1}
\sum_{i+j+k=n\atop i,j,k\geq 0}\omega_i(\omega_j(x,y),\alpha_k(z))-\omega_i(\alpha_j(x),\omega_k(y,z))-\omega_i(\omega_j(y,x),\alpha_k(z))+\omega_i(\alpha_j(y),\omega_k(x,z))=0.
\end{equation}
Moreover, we have
\begin{eqnarray}\label{n-sum-cocycle-1}
&&\sum_{i+j+k=n\atop 0<i,j,k\leq n-1 }\omega_i(\omega_j(x,y),\alpha_k(z))-\omega_i(\alpha_j(x),\omega_k(y,z))-\omega_i(\omega_j(y,x),\alpha_k(z))+\omega_i(\alpha_j(y),\omega_k(x,z))\\
\nonumber&=&(\partial_{\omega\omega}\omega_n+\partial_{\alpha\omega}\alpha_n)(x,y,z).
\end{eqnarray}
For all $x,y\in A$ and $n=1,2,\dots$,  we have
\begin{equation}\label{cocycle-2}
\sum_{i+j+k=n\atop i,j,k\geq 0}\omega_i(\alpha_j(x),\alpha_k(y))-\sum_{i+j=n\atop i,j\geq 0}\alpha_i(\omega_j(x,y))=0.
\end{equation}
Moreover, we have
\begin{equation}\label{n-sum-cocycle-2}
\sum_{i+j+k=n\atop 0<i,j,k\leq n-1}\omega_i(\alpha_j(x),\alpha_k(y))-\sum_{i+j=n\atop i,j>0}\alpha_i(\omega_j(x,y))=(\partial_{\omega\alpha}\omega_n+\partial_{\alpha\alpha}\alpha_n)(x,y).
\end{equation}
\begin{pro}\label{pro:2-cocycle-regular-rep}
Let $(\omega_t=\omega+\sum_{i=1}^{+\infty} \omega_it^i,\alpha_t=\alpha+\sum_{i=1}^{+\infty} \alpha_it^i)$ be a $1$-parameter formal deformation of a Hom-pre-Lie algebra $(A,\omega,\alpha)$. Then $(\omega_1,\alpha_1)$ is a $2$-cocycle of the Hom-pre-Lie algebra $(A,\omega,\alpha)$ with  coefficients in the regular representation.
\end{pro}
\begin{proof}
When $n=1$, by \eqref{cocycle-1}, we have
\begin{eqnarray*}
\nonumber 0&=&(x\cdot y)\cdot\alpha_1(z)+\omega_1(x,y)\cdot\alpha(z)+\omega_1(x\cdot y,\alpha(z))-\alpha(x)\cdot\omega_1(y,z)\\
\nonumber &&-\alpha_1(x)\cdot (y\cdot z)-\omega_1(\alpha(x),y\cdot z)-(y\cdot x)\cdot\alpha_1(z)-\omega_1(y,x)\cdot\alpha(z)\\
\nonumber &&-\omega_1(y\cdot x,\alpha(z))+\alpha(y)\cdot\omega_1(x,z)+\alpha_1(y)\cdot(x\cdot z)+\omega_1(\alpha(y),x\cdot z)\\
&=&-(\partial_{\omega\omega}\omega_1+\partial_{\alpha\omega}\alpha_1)(x,y,z),
\end{eqnarray*}
and by \eqref{cocycle-2}, we have
\begin{eqnarray*}
\nonumber 0&=&\omega_1(\alpha(x),\alpha(y))+\alpha_1(x)\cdot \alpha(y)+\alpha(x)\cdot \alpha_1(y)-\alpha(\omega_1(x,y))-\alpha_1(x\cdot y)\\
&=&-(\partial_{\omega\alpha}\omega_1+\partial_{\alpha\alpha}\alpha_1)(x,y),
\end{eqnarray*}
which implies that $\tilde{\partial}_{\reg}(\omega_1,\alpha_1)=0$. Thus, $(\omega_1,\alpha_1)$ is a $2$-cocycle of the Hom-pre-Lie algebra $(A,\omega,\alpha)$.
\end{proof}

\begin{defi}
The $2$-cocycle $(\omega_1,\alpha_1)$ is called the {\bf infinitesimal} of the $1$-parameter formal deformation $(A[[t]],\omega_t,\alpha_t)$ of the Hom-pre-Lie algebra $(A,\omega,\alpha)$.
\end{defi}

\begin{defi}
Let $(\omega_t'=\omega+\sum_{i=1}^{+\infty} \omega_i't^i,\alpha_t'=\alpha+\sum_{i=1}^{+\infty} \alpha_i't^i)$ and $(\omega_t=\omega+\sum_{i=1}^{+\infty} \omega_it^i,\alpha_t=\alpha+\sum_{i=1}^{+\infty} \alpha_it^i)$ be two $1$-parameter formal deformations of a Hom-pre-Lie algebra $(A,\omega,\alpha)$. A {\bf formal isomorphism} from $(A[[t]],\omega_t',\alpha_t')$ to $(A[[t]],\omega_t,\alpha_t)$ is a power series $\Phi_t=\sum_{i=0}^{+\infty} \varphi_it^i$, where $\varphi_i:A\longrightarrow A$ are linear maps with $\varphi_0={\Id}$, such that
\begin{eqnarray}
 \Phi_t\circ \omega_t'&=&\omega_t \circ(\Phi_t\times \Phi_t),\\
   \alpha_t\circ \Phi_t&=&\Phi_t\circ \alpha_t'.
\end{eqnarray}
Two $1$-parameter formal deformations $(A[[t]],\omega_t',\alpha_t')$ and $(A[[t]],\omega_t,\alpha_t)$ are said to be {\bf equivalent} if there exists a
formal isomorphism $\Phi_t=\sum_{i=0}^{+\infty} \varphi_it^i$ from $(A[[t]],\omega_t',\alpha_t')$ to $(A[[t]],\omega_t,\alpha_t)$.
\end{defi}

\begin{thm}\label{thm:iso3} Let $(A,\omega,\alpha)$ be a Hom-pre-Lie algebra.
If two $1$-parameter formal deformations $(\omega_t'=\omega+\sum_{i=1}^{+\infty} \omega_i't^i,\alpha_t'=\alpha+\sum_{i=1}^{+\infty} \alpha_i't^i)$ and $(\omega_t=\omega+\sum_{i=1}^{+\infty} \omega_it^i,\alpha_t=\alpha+\sum_{i=1}^{+\infty} \alpha_it^i)$ are equivalent, then there infinitesimals $(\omega_1',\alpha_1')$ and $(\omega_1,\alpha_1)$ are in the same cohomology class of $\tilde{H}^2(A;A)$.
\end{thm}
\begin{proof}
Let $(\omega_t',\alpha_t')$ and $(\omega_t,\alpha_t)$ be two $1$-parameter formal deformations. By Proposition \ref{pro:2-cocycle-regular-rep}, we have $(\omega_1',\alpha_1')$ and $(\omega_1,\alpha_1) \in \tilde{Z}^2(A;A)$. Let  $\Phi_t=\sum_{i=0}^{+\infty} \varphi_it^i$ be the formal isomorphism. Then for all $x,y\in A$, we have
\begin{eqnarray*}
\omega_t'(x,y)
&=&\Phi_t^{-1}\circ\omega_t(\Phi_t(x),\Phi_t(y))\\
&=&({\Id}-\varphi_1 t+\dots)\omega_t\big(x+\varphi_1(x)t+\dots,y+\varphi_1(y)t+\dots\big)\\
&=&({\Id}-\varphi_1 t+\dots)\Big(x\cdot y+\big(x\cdot\varphi_1(y)+\varphi_1(x)\cdot y+\omega_1(x,y)\big)t+\dots\Big)\\
&=&x\cdot y+\Big(x\cdot\varphi_1(y)+\varphi_1(x)\cdot y+\omega_1(x,y)-\varphi_1(x\cdot y)\Big)t+\dots.
\end{eqnarray*}
Thus, we have
\begin{eqnarray*}
\omega_1'(x,y)-\omega_1(x,y)&=&x\cdot\varphi_1(y)+\varphi_1(x)\cdot y-\varphi_1(x\cdot y)\\
&=&\partial_{\omega\omega}\varphi_1(x,y),
\end{eqnarray*}
which implies that $\omega_1'-\omega_1=\partial_{\omega\omega}\varphi_1.$

For all $x\in A$, we have
\begin{eqnarray*}
\alpha_t'(x)
&=&\Phi_t^{-1}\circ\alpha_t(\Phi_t(x))\\
&=&({\Id}-\varphi_1 t+\dots)\alpha_t\big(x+\varphi_1(x)t+\dots\big)\\
&=&({\Id}-\varphi_1 t+\dots)\Big(\alpha(x)+\big(\alpha(\varphi_1(x))+\alpha_1(x))t+\dots\Big)\\
&=&\alpha(x)+\Big(\alpha(\varphi_1(x))+\alpha_1(x)-\varphi_1(\alpha(x))\Big)t+\dots.
\end{eqnarray*}
Thus, we have
\begin{equation}
\nonumber\alpha_1'(x)-\alpha_1(x)=\alpha(\varphi_1(x))-\varphi_1(\alpha(x))=\partial_{\omega\alpha}\varphi_1(x),
\end{equation}
which implies that $\alpha_1'-\alpha_1=\partial_{\omega\alpha}\varphi_1.$

Thus, we have $(\omega_1'-\omega_1,\alpha_1'-\alpha_1)\in \tilde{B}^2(A;A)$. This finishes the proof.
\end{proof}

\begin{defi}
A $1$-parameter formal deformation $(A[[t]],\omega_t,\alpha_t)$ of a Hom-pre-Lie algebra $(A,\omega,\alpha)$ is said to be {\bf trivial} if it is equivalent to $(A,\omega,\alpha)$, i.e. there exists $\Phi_t=\sum_{i=0}^{+\infty} \varphi_it^i$, where $\varphi_i:A\longrightarrow A$ are linear maps with $\varphi_0={\Id}$, such that
\begin{eqnarray}
 \Phi_t\circ \omega_t&=&\omega \circ(\Phi_t\times \Phi_t),\\
 \alpha \circ\Phi_t&=&\Phi_t\circ\alpha_t.
\end{eqnarray}

\end{defi}
\begin{defi}
Let $(A,\omega,\alpha)$ be a Hom-pre-Lie algebra. If all  $1$-parameter formal deformations are trivial, we say that $(A,\omega,\alpha)$ is {\bf rigid}.
\end{defi}

\begin{thm}
Let $(A,\omega,\alpha)$ be a Hom-pre-Lie algebra. If $\tilde{H}^2(A;A)=0$, then $(A,\omega,\alpha)$ is rigid.
\end{thm}
\begin{proof}
Let $(\omega_t=\omega+\sum_{i=1}^{+\infty} \omega_it^i,\alpha_t=\alpha+\sum_{i=1}^{+\infty} \alpha_it^i)$ be a $1$-parameter formal deformation and assume that $n\geq1$ is the minimal number such that at least one of $\omega_n$ and $\alpha_n$ is not zero. By \eqref{n-sum-cocycle-1}, \eqref{n-sum-cocycle-2} and $\tilde{H}^2(A;A)=0$, we have $(\omega_n,\alpha_n)\in \tilde{B}^2(A;A)$. Thus, there exists $\varphi_n \in \tilde{\huaC}^1(A;A)$ such that $\omega_n=\partial_{\omega\omega}(-\varphi_n)$ and $\alpha_n=\partial_{\omega\alpha}(-\varphi_n)$. Let $\Phi_t={\Id}+\varphi_nt^n$ and define a new formal deformation $(\omega_t',\alpha_t')$ by $\omega_t'(x,y)=\Phi_t^{-1}\circ\omega_t(\Phi_t(x),\Phi_t(y))$ and $\alpha_t'(x)=\Phi_t^{-1}\circ\alpha_t(\Phi_t(x))$. Then $(\omega_t',\alpha_t')$ and $(\omega_t,\alpha_t)$ are equivalent. By straightforward computation, for all $x,y\in A$, we have
\begin{eqnarray*}
\omega_t'(x,y)
&=&\Phi_t^{-1}\circ\omega_t(\Phi_t(x),\Phi_t(y))\\
&=&({\Id}-\varphi_n t^n+\dots)\omega_t\big(x+\varphi_n(x)t^n,y+\varphi_n(y)t^n\big)\\
&=&({\Id}-\varphi_n t^n+\dots)\Big(x\cdot y+\big(x\cdot\varphi_n(y)+\varphi_n(x)\cdot y+\omega_n(x,y)\big)t^n+\dots\Big)\\
&=&x\cdot y+\Big(x\cdot\varphi_n(y)+\varphi_n(x)\cdot y+\omega_n(x,y)-\varphi_n(x\cdot y)\Big)t^n+\dots.
\end{eqnarray*}
Thus, we have $\omega_1'=\omega_2'=\dots=\omega_{n-1}'=0$. Moreover, we have
\begin{eqnarray*}
\omega_n'(x,y)&=&\varphi_n(x)\cdot y+x\cdot \varphi_n(y)+\omega_n(x,y)-\varphi_n(x\cdot y)\\
&=&\partial_{\omega\omega}\varphi_n(x,y)+\omega_n(x,y)\\
&=&0.
\end{eqnarray*}
For all $x\in A$, we have
\begin{eqnarray*}
\alpha_t'(x)
&=&\Phi_t^{-1}\circ\alpha_t(\Phi_t(x))\\
&=&({\Id}-\varphi_n t^n+\dots)\alpha_t\big(x+\varphi_n(x)t^n\big)\\
&=&({\Id}-\varphi_1 t^n+\dots)\Big(\alpha(x)+\big(\alpha(\varphi_n(x))+\alpha_n(x))t^n+\dots\Big)\\
&=&\alpha(x)+\Big(\alpha(\varphi_n(x))+\alpha_n(x)-\varphi_n(\alpha(x))\Big)t^n+\dots.
\end{eqnarray*}
Thus, we have $\alpha_1'=\alpha_2'=\dots=\alpha_{n-1}'=0$. Moreover, we have
\begin{eqnarray*}
\alpha_n'(x,y)&=&\alpha(\varphi_n(x))+\alpha_n(x)-\varphi_n(\alpha(x))\\
&=&\partial_{\omega\alpha}\varphi_n(x)+\alpha_n(x)\\
&=&0.
\end{eqnarray*}
Keep repeating the process, we obtain that $(A[[t]],\omega_t,\alpha_t)$ is equivalent to $(A,\omega,\alpha)$. The proof is finished.
\end{proof}
At the end of this section, we recall  $1$-parameter formal deformations of Hom-Lie algebras, and establish the relation between  $1$-parameter formal deformations of Hom-pre-Lie algebras and  $1$-parameter formal deformations of Hom-Lie algebras.
\begin{defi}{\rm(\cite{Bene})}
Let $(\g,[\cdot,\cdot]_{\g},\phi_{\g})$ be a Hom-Lie algebra, $[\cdot,\cdot]_t=[\cdot,\cdot]_{\g}+\sum_{i=1}^{+\infty} \bar{\omega}_it^i:\g[[t]]\wedge \g[[t]]\longrightarrow \g[[t]]$ be a $\mathbb K[[t]]$-bilinear map and $\phi_t=\phi_{\g}+\sum_{i=1}^{+\infty} \phi_it^i:\g[[t]]\longrightarrow \g[[t]]$ be a $\mathbb K[[t]]$-linear map, where $\bar{\omega}_i:\g\otimes \g\longrightarrow \g$ and $\phi_i:\g\longrightarrow \g$ are linear maps. If $(\g[[t]],[\cdot,\cdot]_t,\phi_t)$ is still a Hom-Lie algebra, we say that $\{\bar{\omega}_i,\phi_i\}_{i\geq1}$ generates a {\bf $1$-parameter formal deformation} of a Hom-Lie algebra $(\g,[\cdot,\cdot]_{\g},\phi_{\g})$.
\end{defi}

\begin{pro}
Let $(\omega_t=\omega+\sum_{i=1}^{+\infty} \omega_it^i,\alpha_t=\alpha+\sum_{i=1}^{+\infty} \alpha_it^i)$ be a $1$-parameter formal deformation of a Hom-pre-Lie algebra $(A,\omega,\alpha)$, then $$\{\omega_i-\omega_i\circ \sigma,\alpha_i\}_{i\geq1}$$ generates a $1$-parameter formal deformation of the sub-adjacent Hom-Lie algebra $A^C$, where $\sigma:A\otimes A\longrightarrow A\otimes A$ is the flip operator defined by $\sigma(x\otimes y)=y\otimes x$ for all $x,y\in A$.
\end{pro}
\begin{proof}
Let $(A[[t]],\omega_t,\alpha_t)$ be a  $1$-parameter formal deformation of a Hom-pre-Lie algebra $(A,\omega,\alpha)$. For all $x,y\in A$, we have
\begin{eqnarray*}
\omega_t(x,y)-\omega_t(y,x)&=&\omega(x,y)-\omega(y,x)+\sum_{i=1}^{+\infty} \omega_i(x,y)t^i-\sum_{i=1}^{+\infty} \omega_i(y,x)t^i\\
&=&[x,y]_C+\sum_{i=1}^{+\infty}(\omega_i-\omega_i\circ \sigma)(x,y)t^i,
\end{eqnarray*}
and
\begin{equation}
\nonumber\alpha_t\circ(\omega_t(x,y)-\omega_t(y,x))=\omega_t(\alpha_t(x),\alpha_t(y))-\omega_t(\alpha_t(y),\alpha_t(x)).
\end{equation}
Therefore, $\{\omega_i-\omega_i\circ \sigma,\alpha_i\}_{i\geq1}$ generates a $1$-parameter formal deformation of the sub-adjacent Hom-Lie algebra $A^C$.
\end{proof}

\section{Abelian extensions of Hom-pre-Lie algebras}\label{sec:ext}
In this section, we study abelian extensions of Hom-pre-Lie algebras using the cohomological approach. We show that abelian extensions are classified by the cohomology group $\tilde{H}^2(A;V)$.

\begin{defi}
Let $(A,\cdot,\alpha)$ and $(V,\cdot_V,\beta)$ be two Hom-pre-Lie algebras. An {\bf   extension} of $(A,\cdot,\alpha)$ by $(V,\cdot_V,\beta)$ is a short exact sequence of Hom-pre-Lie algebra morphisms:
$$\xymatrix{
  0 \ar[r] &V \ar[d]_{\beta}\ar[r]^{\iota}& \hat{A}\ar[d]_{\alpha_{\hat{A}}}\ar[r]^{p}&A\ar[d]_{\alpha}\ar[r]&0\\
     0\ar[r] &V \ar[r]^{ \iota} &\hat{A}\ar[r]^{p} &A\ar[r]&0,              }$$
where $(\hat{A},\cdot_{\hat{A}},\alpha_{\hat{A}})$ is a Hom-pre-Lie algebra.

It is called an {\bf abelian extension} if  $(V,\cdot_V,\beta)$ is an abelian Hom-pre-Lie algebra, i.e.   for all $u,v\in V, u\cdot_V v=0$.
\end{defi}
\begin{defi}
A {\bf section} of an extension $(\hat{A},\cdot_{\hat{A}},\alpha_{\hat{A}})$ of a Hom-pre-Lie algebra $(A,\cdot,\alpha)$ by  $(V,\cdot_V,\beta)$ is a linear map $s:A\longrightarrow \hat{A}$ such that $p\circ s=\Id_A$.
\end{defi}

Let $(\hat{A},\cdot_{\hat{A}},\alpha_{\hat{A}})$ be an abelian extension of a Hom-pre-Lie algebra $(A,\cdot,\alpha)$ by $(V,\beta)$ and $s:A\longrightarrow \hat{A}$ a section. For all $x,y\in A$, define linear maps $\theta:A\otimes A\longrightarrow V$ and $\xi:A\longrightarrow V$ respectively by
\begin{eqnarray}
\label{product}\theta(x,y)&=&s(x)\cdot_{\hat{A}}s(y)-s(x\cdot y),\\
\label{morphism}\xi(x)&=&\alpha_{\hat{A}}(s(x))-s(\alpha(x)).
\end{eqnarray}
And for all $x,y\in A, u\in V$, define $\rho,\mu:A\longrightarrow\gl(V)$ respectively by
\begin{eqnarray}
\label{Hom-pre-Lie-representation-1}\rho(x)(u)&=&s(x)\cdot_{\hat{A}} u,\\
\label{Hom-pre-Lie-representation-2}\mu(x)(u)&=&u\cdot_{\hat{A}} s(x).
\end{eqnarray}
Obviously, $\hat{A}$ is isomorphic to $A\oplus V$ as vector spaces. Transfer the Hom-pre-Lie algebra structure on $\hat{A}$ to that on $A\oplus V$, we obtain a Hom-pre-Lie algebra $(A\oplus V,\diamond,\phi)$, where $\diamond$ and $\phi$ are given by
\begin{eqnarray}
  \label{eq:6.1}(x+u)\diamond (y+v)&=&x\cdot y+\theta(x,y)+\rho(x)(v)+\mu(y)(u),\quad \forall ~x,y\in A, u,v\in V,\\
   \label{eq:6.2}\phi(x+u)&=&\alpha(x)+\xi(x)+\beta(u), \quad \forall ~x\in A, u\in V.
\end{eqnarray}

\begin{thm}\label{thm:cocycle}
With the above notations, we have
\begin{itemize}
\item[$\rm(i)$] $(V,\beta,\rho,\mu)$ is a representation of the Hom-pre-Lie algebra $(A,\cdot,\alpha)$,
\item[$\rm(ii)$] $(\theta,\xi)$ is a $2$-cocycle of the Hom-pre-Lie algebra $(A,\cdot,\alpha)$ with coefficients in the representation $(V,\beta,\rho,\mu)$.
\end{itemize}
\end{thm}

\begin{proof}
For all $x\in A$, $v\in V$, by the definition of a Hom-pre-Lie algebra, we have
\begin{equation}
\nonumber\phi(x\diamond v)-\phi(x)\diamond \phi(v)=\beta(\rho(x)(v))-\rho(\alpha(x))\beta(v)=0,
\end{equation}
which implies that
\begin{equation}\label{representation-1}
\beta\circ \rho(x)=\rho(\alpha(x))\circ\beta.
\end{equation}
Similarly, we have
\begin{equation}\label{representation-2}
\beta\circ \mu(x)=\mu(\alpha(x))\circ\beta.
\end{equation}
For all $x,y\in A$, $v\in V$, by the definition of a Hom-pre-Lie algebra, we have
\begin{eqnarray}\nonumber&&(x\diamond y)\diamond\phi(v)-\phi(x)\diamond (y\diamond v)-(y\diamond x)\diamond\phi(v)+\phi(y)\diamond (x\diamond v)\\
\nonumber &=&(x\cdot y+\theta(x,y))\diamond \beta(v)-(\alpha(x)+\xi(x))\diamond\rho(y)(v)\\
\nonumber &&-(y\cdot x+\theta(y,x))\diamond \beta(v)+(\alpha(y)+\xi(y))\diamond\rho(x)(v)\\
\nonumber &=&\rho(x\cdot y)\beta(v)-\rho(\alpha(x))\rho(y)(v)-\rho(y\cdot x)\beta(v)+\rho(\alpha(y))\rho(x)(v)\\
\nonumber &=&\rho([x,y]_C)\beta(v)-\rho(\alpha(x))\rho(y)(v)+\rho(\alpha(y))\rho(x)(v)\\
\nonumber&=&0,
\end{eqnarray}
which implies that
\begin{equation}\label{representation-3}
\rho([x,y]_C)\circ\beta=\rho(\alpha(x))\circ\rho(y)-\rho(\alpha(y))\circ\rho(x).
\end{equation}
Similarly, we have
\begin{equation}\label{representation-4}
\mu(\alpha(y))\circ\mu(x)-\mu(x\cdot y)\circ \beta=\mu(\alpha(y))\circ\rho(x)-\rho(\alpha(x))\circ\mu(y).
\end{equation}
By \eqref{representation-1}, \eqref{representation-2}, \eqref{representation-3},  \eqref{representation-4}, we obtain that $(V,\beta,\rho,\mu)$ is a representation.

For all $x,y\in A$, by the definition of a Hom-pre-Lie algebra, we have
\begin{eqnarray}\nonumber&&\phi(x\diamond y)-\phi(x)\diamond \phi(y)\\
\nonumber &=&\phi\big(x\cdot y+\theta(x,y)\big)-\alpha(x)\cdot \alpha(y)-\theta(\alpha(x),\alpha(y))-\rho(\alpha(x))\xi(y)-\mu(\alpha(y))\xi(x)\\
\nonumber &=&\xi(x\cdot y)+\beta(\theta(x,y))-\theta(\alpha(x),\alpha(y))-\rho(\alpha(x))\xi(y)-\mu(\alpha(y))\xi(x)\\
\nonumber &=&\partial_{\omega\alpha}\theta(x,y)+\partial_{\alpha\alpha}\xi(x,y)\\
\nonumber &=&0,
\end{eqnarray}
which implies that
\begin{equation}\label{2-cocycle-1}
\partial_{\omega\alpha}\theta+\partial_{\alpha\alpha}\xi=0.
\end{equation}
For all $x,y,z\in A$, by the definition of a Hom-pre-Lie algebra, we have
\begin{eqnarray}\nonumber&&(x\diamond y)\diamond\phi(z)-\phi(x)\diamond (y\diamond z)-(y\diamond x)\diamond\phi(z)+\phi(y)\diamond (x\diamond z)\\
\nonumber &=&\big(x\cdot y+\theta(x,y)\big)\diamond(\alpha(z)+\xi(z))-(\alpha(x)+\xi(x))\diamond\big(y\cdot z +\theta(y,z)\big)\\
\nonumber &&-\big(y\cdot x+\theta(y,x)\big)\diamond(\alpha(z)+\xi(z))+(\alpha(y)+\xi(y))\diamond\big(x\cdot z +\theta(x,z)\big)\\
\nonumber &=&\theta(x\cdot y,\alpha(z))+\mu(\alpha(z))\theta(x,y)-\theta(\alpha(x),y\cdot z)-\rho(\alpha(x))\theta(y,z)\\
\nonumber &&-\theta(y\cdot x,\alpha(z))-\mu(\alpha(z))\theta(y,x)+\theta(\alpha(y),x\cdot z)+\rho(\alpha(y))\theta(x,z)\\
\nonumber &&+\rho(x\cdot y)\xi(z)-\mu(y\cdot z)\xi(x)-\rho(y\cdot x)\xi(z)+\mu(x\cdot z)\xi(y)\\
\nonumber &=&-\partial_{\omega\omega}\theta(x,y,z)-\partial_{\alpha\omega}\xi(x,y,z)\\
\nonumber &=&0,
\end{eqnarray}
which implies that
\begin{equation}\label{2-cocycle-2}
\partial_{\omega\omega}\theta+\partial_{\alpha\omega}\xi=0.
\end{equation}
By \eqref{2-cocycle-1} and \eqref{2-cocycle-2}, we obtain that $\tilde{\partial}(\theta,\xi)=0$, which implies that $(\theta,\xi)$ is a $2$-cocycle of the Hom-pre-Lie algebra $(A,\cdot,\alpha)$. The proof is finished.
\end{proof}
\begin{defi}\label{defi:morphism}
Let $(\hat{A_1},\cdot_{\hat{A_1}},\alpha_{\hat{A_1}})$ and $(\hat{A_2},\cdot_{\hat{A_2}},\alpha_{\hat{A_2}})$ be two abelian extensions of a Hom-pre-Lie algebra $(A,\cdot,\alpha)$ by  $(V,\beta)$. They are said to be {\bf isomorphic} if there exists a Hom-pre-Lie algebra isomorphism $\zeta:(\hat{A_1},\cdot_{\hat{A_1}},\alpha_{\hat{A_1}})\longrightarrow (\hat{A_2},\cdot_{\hat{A_2}},\alpha_{\hat{A_2}})$ such that the following diagram is commutative:
$$\xymatrix{
  0 \ar[r] &V\ar @{=}[d]\ar[r]^{\iota_1}& \hat{A_1}\ar[d]_{\zeta}\ar[r]^{p_1}&A\ar @{=}[d]\ar[r]&0\\
     0\ar[r] &V\ar[r]^{ \iota_2} &\hat{A_2}\ar[r]^{p_2} &A\ar[r]&0.              }$$
\end{defi}

\begin{pro}
With the above notations, we have
\begin{itemize}
\item[$\rm(i)$] Two different sections of an abelian extension of a Hom-pre-Lie algebra $(A,\cdot,\alpha)$ by $(V,\beta)$ give rise to the same representation of $(A,\cdot,\alpha)$,
\item[$\rm(ii)$] Isomorphic abelian extensions give rise to the same representation of $(A,\cdot,\alpha)$.
\end{itemize}
\end{pro}
\begin{proof}
$\rm(i)$ Let $(\hat{A},\cdot_{\hat{A}},\alpha_{\hat{A}})$ be an abelian extension of a Hom-pre-Lie algebra $(A,\cdot,\alpha)$ by  $(V,\beta)$. Choosing two different sections $s_1,s_2:A\longrightarrow \hat{A}$, by \eqref{Hom-pre-Lie-representation-1}, \eqref{Hom-pre-Lie-representation-2} and Theorem \ref{thm:cocycle}, we obtain two representations $(V,\beta,\rho_1,\mu_1)$ and $(V,\beta,\rho_2,\mu_2)$. Define $\varphi:A\longrightarrow V$ by $\varphi(x)=s_1(x)-s_2(x)$. Then for all $x\in A$, we have
\begin{eqnarray*}
 \rho_1(x)(u)-\rho_2(x)(u)&=&s_1(x)\cdot_{\hat{A}} u-s_2(x)\cdot_{\hat{A}} u\\
  &=&(\varphi(x)+s_2(x))\cdot_{\hat{A}} u-s_2(x)\cdot_{\hat{A}} u\\
  &=&\varphi(x)\cdot_{\hat{A}} u\\
  &=&0,
\end{eqnarray*}
which implies that $\rho_1=\rho_2$.
Similarly, we have $\mu_1=\mu_2$. This finishes the proof.

$\rm(ii)$ Let $(\hat{A_1},\cdot_{\hat{A_1}},\alpha_{\hat{A_1}})$ and $(\hat{A_2},\cdot_{\hat{A_2}},\alpha_{\hat{A_2}})$ are two isomorphic abelian extensions of a Hom-pre-Lie algebra $(A,\cdot,\alpha)$ by $(V,\beta)$. Let $s_1:A_1\longrightarrow \hat{A_1}$ and $s_2:A_2\longrightarrow \hat{A_2}$ be two sections of $(\hat{A_1},\cdot_{\hat{A_1}},\alpha_{\hat{A_1}})$ and $(\hat{A_2},\cdot_{\hat{A_2}},\alpha_{\hat{A_2}})$ respectively. By \eqref{Hom-pre-Lie-representation-1}, \eqref{Hom-pre-Lie-representation-2} and Theorem \ref{thm:cocycle}, we obtain that $(V,\beta,\rho_1,\mu_1)$ and $(V,\beta,\rho_2,\mu_2)$ are their representations respectively. Define $s'_1:A_1\longrightarrow \hat{A_1}$ by $s'_1=\zeta^{-1}\circ s_2$. Since $\zeta:(\hat{A_1},\cdot_{\hat{A_1}},\alpha_{\hat{A_1}})\longrightarrow (\hat{A_2},\cdot_{\hat{A_2}},\alpha_{\hat{A_2}})$ is a Hom-pre-Lie algebra isomorphism satisfying the commutative diagram in Definition \ref{defi:morphism}, by $p_2\circ \zeta=p_1$, we have
\begin{equation}
p_1\circ s'_1=p_2\circ \zeta \circ \zeta^{-1}\circ s_2=\Id_A.
\end{equation}
Thus, we obtain that $s'_1$ is a section of $(\hat{A_1},\cdot_{\hat{A_1}},\alpha_{\hat{A_1}})$. For all $x\in A, u\in V$, we have
\begin{eqnarray*}
\rho_1(x)(u)&=&s'_1(x) \cdot_{\hat{A_1}} u\\
  &=&(\zeta^{-1}\circ s_2)(x)\cdot_{\hat{A_1}} u\\
  &=&\zeta^{-1}(s_2(x) \cdot_{\hat{A_2}} u)\\
  &=&\rho_2(x)(u),
\end{eqnarray*}
which implies that $\rho_1=\rho_2$.
Similarly, we have $\mu_1=\mu_2$. This finishes the proof.
\end{proof}
So in the sequel, we fix a representation $(V,\beta,\rho,\mu)$ of a Hom-pre-Lie algebra $(A,\cdot,\alpha)$ and consider abelian extensions that induce the given representation.
\begin{thm}
Abelian extensions of a Hom-pre-Lie algebra $(A,\cdot,\alpha)$ by $(V,\beta)$ are classified by $\tilde{H}^2(A;V)$.
\end{thm}
\begin{proof}
Let $(\hat{A},\cdot_{\hat{A}},\alpha_{\hat{A}})$ be an abelian extension of a Hom-pre-Lie algebra $(A,\cdot,\alpha)$ by  $(V,\beta)$. Choosing a section $s:A\longrightarrow \hat{A}$, by Theorem \ref{thm:cocycle}, we obtain that $(\theta,\xi)\in \tilde{Z}^2(A;V)$. Now we show that the cohomological class of $(\theta,\xi)$ does not depend on the choice of sections. In fact, let $s_1$ and $s_2$ be two different sections. Define $\varphi:A\longrightarrow V$ by $\varphi(x)=s_1(x)-s_2(x)$. Then for all $x,y\in A$, we have
\begin{eqnarray*}
 \theta_1(x,y)&=&s_1(x)\cdot_{\hat{A}} s_1(y)-s_1(x\cdot y)\\
  &=&\big(s_2(x)+\varphi(x)\big)\cdot_{\hat{A}} \big(s_2(y)+\varphi(y)\big)-s_2(x\cdot y)-\varphi(x\cdot y)\\
  &=&s_2(x)\cdot_{\hat{A}} s_2(y)+\rho(x)\varphi(y)+\mu(y)\varphi(x)-s_2(x\cdot y)-\varphi(x\cdot y)\\
  &=&\theta_2(x,y)+\partial_{\omega\omega}\varphi(x,y),
\end{eqnarray*}
which implies that $\theta_1-\theta_2=\partial_{\omega\omega}\varphi$.

For all $x\in A$, we have
\begin{eqnarray*}
 \xi_1(x)&=&\alpha_{\hat{A}}(s_1(x))-s_1(\alpha(x))\\
  &=&\alpha_{\hat{A}}(\varphi(x)+s_2(x))-\varphi(\alpha(x))-s_2(\alpha(x))\\
  &=&\alpha_{\hat{A}}(\varphi(x))+\alpha_{\hat{A}}(s_2(x))-\varphi(\alpha(x))-s_2(\alpha(x))\\
  &=&\xi_2(x)+\beta(\varphi(x))-\varphi(\alpha(x)),
\end{eqnarray*}
which implies that $\xi_1-\xi_2=\partial_{\omega\alpha}\varphi$.

Therefore, we obtain that $(\theta_1-\theta_2,\xi_1-\xi_2)\in \tilde{B}^2(A;V)$, $(\theta_1,\xi_1)$ and $(\theta_2,\xi_2)$ are in the same cohomological class.

Now we  prove that isomorphic abelian extensions give rise to the same element in  $\tilde{H}^2(A;V)$. Assume that $(\hat{A_1},\cdot_{\hat{A_1}},\alpha_{\hat{A_1}})$ and $(\hat{A_2},\cdot_{\hat{A_2}},\alpha_{\hat{A_2}})$ are two isomorphic abelian extensions of a Hom-pre-Lie algebra $(A,\cdot,\alpha)$ by $(V,\beta)$, and $\zeta:(\hat{A_1},\cdot_{\hat{A_1}},\alpha_{\hat{A_1}})\longrightarrow (\hat{A_2},\cdot_{\hat{A_2}},\alpha_{\hat{A_2}})$ is a Hom-pre-Lie algebra isomorphism satisfying the commutative diagram in Definition \ref{defi:morphism}. Assume that $s_1:A\longrightarrow \hat{A_1}$ is a section of $\hat{A_1}$. By $p_2\circ \zeta=p_1$, we have
\begin{equation}
p_2\circ (\zeta\circ s_1)=p_1\circ s_1=\Id_A.
\end{equation}
Thus, we obtain that $\zeta\circ s_1$ is a section of $\hat{A_2}$. Define $s_2=\zeta\circ s_1$. Since $\zeta$ is an isomorphism of Hom-pre-Lie algebras and $\zeta\mid_V=\Id_V$, for all $x,y\in A$, we have
\begin{eqnarray*}
 \theta_2(x,y)&=&s_2(x)\cdot_{\hat{A_2}} s_2(y)-s_2(x\cdot y)\\
  &=&(\zeta\circ s_1)(x)\cdot_{\hat{A_2}}(\zeta\circ s_1)(y)-(\zeta\circ s_1)(x\cdot y)\\
  &=&\zeta\big(s_1(x)\cdot_{\hat{A_1}} s_1(y)-s_1(x\cdot y)\big)\\
  &=&\theta_1(x,y),
\end{eqnarray*}
and
\begin{eqnarray*}
 \xi_2(x)&=&\alpha_{\hat{A_2}}(s_2(x))-s_2(\alpha(x))\\
  &=&\alpha_{\hat{A_2}}(\zeta(s_1(x)))-\zeta(s_1(\alpha(x)))\\
  &=&\zeta(\alpha_{\hat{A_1}}(s_1(x))-s_1(\alpha(x)))\\
  &=&\xi_1(x).
\end{eqnarray*}
Thus, isomorphic abelian extensions gives rise to the same element in $\tilde{H}^2(A;V)$.

Conversely, given two 2-cocycles $(\theta_1,\xi_1)$ and $(\theta_2,\xi_2)$, by \eqref{eq:6.1} and \eqref{eq:6.2}, we can construct two abelian extensions $(A\oplus V,\diamond_1,\phi_1)$ and $(A\oplus V,\diamond_2,\phi_2)$. If  $(\theta_1,\xi_1), (\theta_2,\xi_2)\in \tilde{H}^2(A;V)$, then there exists $\varphi:A\longrightarrow V$, such that $\theta_1=\theta_2+\partial_{\omega\omega}\varphi$ and $\xi_1=\xi_2+\partial_{\omega\alpha}\varphi$. We define $\zeta:A\oplus V\longrightarrow A\oplus V$ by
\begin{equation}
\zeta(x+u)=x+u+\varphi(x),\quad \forall ~x\in A, u\in V.
\end{equation}
For all $x,y\in A, u,v\in V$, by $\theta_1=\theta_2+\partial_{\omega\omega}\varphi$, we have
\begin{eqnarray}\nonumber&&\zeta\big((x+u)\diamond_1(y+v)\big)- \zeta(x+u)\diamond_2\zeta(y+v)\\
\nonumber &=&\zeta\big(x\cdot y+\theta_1(x,y)+\rho(x)(v)+\mu(y)(u)\big)-\big(x+u+\varphi(x)\big)\diamond_2 \big(y+v+\varphi(y)\big)\\
\nonumber &=&\theta_1(x,y)+\varphi(x\cdot y)-\theta_2(x,y)-\rho(x)\varphi(y)-\mu(y)\varphi(x)\\
\nonumber &=&\theta_1(x,y)-\theta_2(x,y)-\partial_{\omega\omega}\varphi(x,y)\\
 &=&0,
\end{eqnarray}
and for all $x\in A, u\in V$, by $\xi_1=\xi_2+\partial_{\omega\alpha}\varphi$, we have
\begin{eqnarray}\nonumber&&\zeta \circ\phi_1(x+u)-\phi_1\circ \zeta(x+u)\\
\nonumber &=&\zeta\big(\alpha(x)+\xi_1(x)+\beta(u)\big)-\phi_2\big(x+u+\varphi(x)\big)\\
\nonumber &=&\xi_1(x)+\varphi(\alpha(x))-\xi_2(x)-\beta(\varphi(x))\\
\nonumber &=&\xi_1(x)-\xi_2(x)-\partial_{\omega\alpha}\varphi(x)\\
 &=&0,
\end{eqnarray}
which implies that $\zeta$ is a Hom-pre-Lie algebra isomorphism from $(A\oplus V,\diamond_1,\phi_1)$ to $(A\oplus V,\diamond_2,\phi_2)$. Moreover, it is obvious that the diagram in Definition \ref{defi:morphism} is commutative. This finishes the proof.
\end{proof}

\section*{Appendix: The proof of Lemma \ref{lem:cohomology-computation}}
By straightforward computations, for all $x_1,\dots,x_{n+2}\in A$, we have
{\footnotesize
\begin{eqnarray}
\nonumber&&\partial_{\omega\omega}(\partial_{\omega\omega}\varphi)(x_1,\dots,x_{n+2})\\
\nonumber&=&\sum_{i=1}^{n+1}(-1)^{i+1}\rho(\alpha^n(x_i))(\partial_{\omega\omega}\varphi)(x_1,\dots,\widehat{x_i},\dots,x_{n+2})\\
\nonumber&&+\sum_{i=1}^{n+1}(-1)^{i+1}\mu(\alpha^n(x_{n+2}))(\partial_{\omega\omega}\varphi)(x_1,\dots,\widehat{x_i},\dots,x_{n+1},x_i)\\
\nonumber&&-\sum_{i=1}^{n+1}(-1)^{i+1} (\partial_{\omega\omega}\varphi)(\alpha(x_1),\dots,\widehat{\alpha(x_i)},\dots,\alpha(x_{n+1}),x_i\cdot x_{n+2})\\
\nonumber&&+\sum_{1\leq i<j\leq n+1}(-1)^{i+j} (\partial_{\omega\omega}\varphi)([x_i,x_j]_C,\alpha(x_1),\dots,\widehat{\alpha(x_i)},\dots,\widehat{\alpha(x_j)},\dots,\alpha(x_{n+2}))\\
\label{eq-9}&=&\sum_{1\leq j<i\leq n+1}(-1)^{i+1}(-1)^{j+1}\rho(\alpha^n(x_i))\rho(\alpha^{n-1}(x_j))\varphi(x_1,\dots,\widehat{x_j},\dots,\widehat{x_i},\dots,x_{n+2})\\
\label{eq-10}&&+\sum_{1\leq i<j\leq n+1}(-1)^{i+1}(-1)^j\rho(\alpha^n(x_i))\rho(\alpha^{n-1}(x_j))\varphi(x_1,\dots,\widehat{x_i},\dots,\widehat{x_j},\dots,x_{n+2})\\
\label{eq-11}&&+\sum_{1\leq j<i\leq n+1}(-1)^{i+1}(-1)^{j+1}\rho(\alpha^n(x_i))\mu(\alpha^{n-1}(x_{n+2}))\varphi(x_1,\dots,\widehat{x_j},\dots,\widehat{x_i},\dots,x_{n+1},x_j)\\
\label{eq-12}&&+\sum_{1\leq i<j\leq n+1}(-1)^{i+1}(-1)^j\rho(\alpha^n(x_i))\mu(\alpha^{n-1}(x_{n+2}))\varphi(x_1,\dots,\widehat{x_i},\dots,\widehat{x_j},\dots,x_{n+1},x_j)\\
\label{eq-13}&&-\sum_{1\leq j<i\leq n+1}(-1)^{i+1}(-1)^{j+1}\rho(\alpha^n(x_i))\varphi(\alpha(x_1),\dots,\widehat{\alpha(x_j)},\dots,\widehat{\alpha(x_i)},\dots,\alpha(x_{n+1}),x_j\cdot x_{n+2})\\
\label{eq-14}&&-\sum_{1\leq i<j\leq n+1}(-1)^{i+1}(-1)^j\rho(\alpha^n(x_i))\varphi(\alpha(x_1),\dots,\widehat{\alpha(x_i)},\dots,\widehat{\alpha(x_j)},\dots,\alpha(x_{n+1}),x_j\cdot x_{n+2})\\
\label{eq-15}&&+\sum_{1\leq j<k<i\leq n+1}(-1)^{i+1}(-1)^{j+k}\rho(\alpha^n(x_i))\varphi([x_j,x_k]_C,\alpha(x_1),\dots,\widehat{\alpha(x_j)},\dots,\widehat{\alpha(x_k)},\dots,\widehat{\alpha(x_i)},\dots,\alpha(x_{n+2}))\\
\label{eq-16}&&+\sum_{1\leq j<i<k\leq n+1}(-1)^{i+1}(-1)^{j+k-1}\rho(\alpha^n(x_i))\varphi([x_j,x_k]_C,\alpha(x_1),\dots,\widehat{\alpha(x_j)},\dots,\widehat{\alpha(x_i)},\dots,
\widehat{\alpha(x_k)},\dots,\alpha(x_{n+2}))\\
\label{eq-17}&&+\sum_{1\leq i<j<k\leq n+1}(-1)^{i+1}(-1)^{j+k}\rho(\alpha^n(x_i))\varphi([x_j,x_k]_C,\alpha(x_1),\dots,\widehat{\alpha(x_i)},\dots,\widehat{\alpha(x_j)},\dots,
\widehat{\alpha(x_k)},\dots,\alpha(x_{n+2}))\\
\label{eq-18}&&+\sum_{1\leq j<i\leq n+1}(-1)^{i+1}(-1)^{j+1}\mu(\alpha^n(x_{n+2}))\rho(\alpha^{n-1}(x_j))\varphi(x_1,\dots,\widehat{x_j},\dots,\widehat{x_i},\dots,x_{n+1},x_i)\\
\label{eq-19}&&+\sum_{1\leq i<j\leq n+1}(-1)^{i+1}(-1)^j\mu(\alpha^n(x_{n+2}))\rho(\alpha^{n-1}(x_j))\varphi(x_1,\dots,\widehat{x_i},\dots,\widehat{x_j},\dots,x_{n+1},x_i)\\
\label{eq-20}&&+\sum_{1\leq j<i\leq n+1}(-1)^{i+1}(-1)^{j+1}\mu(\alpha^n(x_{n+2}))\mu(\alpha^{n-1}(x_i))\varphi(x_1,\dots,\widehat{x_j},\dots,\widehat{x_i},\dots,x_{n+1},x_j)\\
\label{eq-21}&&+\sum_{1\leq i<j\leq n+1}(-1)^{i+1}(-1)^j\mu(\alpha^n(x_{n+2}))\mu(\alpha^{n-1}(x_i))\varphi(x_1,\dots,\widehat{x_i},\dots,\widehat{x_j},\dots,x_{n+1},x_j)\\
\label{eq-22}&&-\sum_{1\leq j<i\leq n+1}(-1)^{i+1}(-1)^{j+1}\mu(\alpha^n(x_{n+2}))\varphi(\alpha(x_1),\dots,\widehat{\alpha(x_j)},\dots,\widehat{\alpha(x_i)},\dots,\alpha(x_{n+1}),x_j\cdot x_i)\\
\label{eq-23}&&-\sum_{1\leq i<j\leq n+1}(-1)^{i+1}(-1)^j\mu(\alpha^n(x_{n+2}))\varphi(\alpha(x_1),\dots,\widehat{\alpha(x_i)},\dots,\widehat{\alpha(x_j)},\dots,\alpha(x_{n+1}),x_j\cdot x_i)\\
\label{eq-24}&&+\sum_{1\leq j<k<i\leq n+1}(-1)^{i+1}(-1)^{j+k}\mu(\alpha^n(x_{n+2}))\varphi([x_j,x_k]_C,\alpha(x_1),\dots,\widehat{\alpha(x_j)},\dots,\widehat{\alpha(x_k)},\\
\nonumber&&\dots,\widehat{\alpha(x_i)},\dots,\alpha(x_{n+1}),\alpha(x_i))\\
\label{eq-25}&&+\sum_{1\leq j<i<k\leq n+1}(-1)^{i+1}(-1)^{j+k-1}\mu(\alpha^n(x_{n+2}))\varphi([x_j,x_k]_C,\alpha(x_1),\dots,\widehat{\alpha(x_j)},\dots,\widehat{\alpha(x_i)},\\
\nonumber&&\dots,\widehat{\alpha(x_k)},\dots,\alpha(x_{n+1}),\alpha(x_i))\\
\label{eq-26}&&+\sum_{1\leq i<j<k\leq n+1}(-1)^{i+1}(-1)^{j+k}\mu(\alpha^n(x_{n+2}))\varphi([x_j,x_k]_C,\alpha(x_1),\dots,\widehat{\alpha(x_i)},\dots,\widehat{\alpha(x_j)},\\
\nonumber&&\dots,\widehat{\alpha(x_k)},\dots,\alpha(x_{n+1}),\alpha(x_i))\\
\label{eq-27}&&-\sum_{1\leq j<i\leq n+1}(-1)^{i+1}(-1)^{j+1}\rho(\alpha^n(x_j))\varphi(\alpha(x_1),\dots,\widehat{\alpha(x_j)},\dots,\widehat{\alpha(x_i)},\dots,\alpha(x_{n+1}),x_i\cdot x_{n+2})\\
\label{eq-28}&&-\sum_{1\leq i<j\leq n+1}(-1)^{i+1}(-1)^j\rho(\alpha^n(x_j))\varphi(\alpha(x_1),\dots,\widehat{\alpha(x_i)},\dots,\widehat{\alpha(x_j)},\dots,\alpha(x_{n+1}),x_i\cdot x_{n+2})\\
\label{eq-29}&&-\sum_{1\leq j<i\leq n+1}(-1)^{i+1}(-1)^{j+1}\mu(\alpha^{n-1}(x_i \cdot x_{n+2}))\varphi(\alpha(x_1),\dots,\widehat{\alpha(x_j)},\dots,\widehat{\alpha(x_i)},\dots,\alpha(x_{n+1}),
\alpha(x_j))\\
\label{eq-30}&&-\sum_{1\leq i<j\leq n+1}(-1)^{i+1}(-1)^j\mu(\alpha^{n-1}(x_i \cdot x_{n+2}))\varphi(\alpha(x_1),\dots,\widehat{\alpha(x_i)},\dots,\widehat{\alpha(x_j)},\dots,\alpha(x_{n+1}),
\alpha(x_j))\\
\label{eq-31}&&+\sum_{1\leq j<i\leq n+1}(-1)^{i+1}(-1)^{j+1}\varphi(\alpha^2(x_1),\dots,\widehat{\alpha^2(x_j)},\dots,\widehat{\alpha^2(x_i)},\dots,\alpha^2(x_{n+1}),
\alpha(x_j)\cdot(x_i\cdot x_{n+2}))\\
\label{eq-32}&&+\sum_{1\leq i<j\leq n+1}(-1)^{i+1}(-1)^j\varphi(\alpha^2(x_1),\dots,\widehat{\alpha^2(x_i)},\dots,\widehat{\alpha^2(x_j)},\dots,\alpha^2(x_{n+1}),
\alpha(x_j)\cdot(x_i\cdot x_{n+2}))\\
\label{eq-33}&&-\sum_{1\leq j<k<i\leq n+1}(-1)^{i+1}(-1)^{j+k}\varphi([\alpha(x_j),\alpha(x_k)]_C,\alpha^2(x_1),\dots,\widehat{\alpha^2(x_j)},\dots,\widehat{\alpha^2(x_k)},\dots,\\
\nonumber&&\widehat{\alpha^2(x_i)},\dots,\alpha^2(x_{n+1}),\alpha(x_i \cdot x_{n+2}))\\
\label{eq-34}&&-\sum_{1\leq j<i<k\leq n+1}(-1)^{i+1}(-1)^{j+k-1}\varphi([\alpha(x_j),\alpha(x_k)]_C,\alpha^2(x_1),\dots,\widehat{\alpha^2(x_j)},\dots,\widehat{\alpha^2(x_i)},\dots,\\
\nonumber&&\widehat{\alpha^2(x_k)},\dots,\alpha^2(x_{n+1}),\alpha(x_i \cdot x_{n+2}))\\
\label{eq-35}&&-\sum_{1\leq i<j<k\leq n+1}(-1)^{i+1}(-1)^{j+k}\varphi([\alpha(x_j),\alpha(x_k)]_C,\alpha^2(x_1),\dots,\widehat{\alpha^2(x_i)},\dots,\widehat{\alpha^2(x_j)},\dots,\\
\nonumber&&\widehat{\alpha^2(x_k)},\dots,\alpha^2(x_{n+1}),\alpha(x_i \cdot x_{n+2}))\\
\label{eq-36}&&+\sum_{1\leq i<j\leq n+1}(-1)^{i+j}\rho(\alpha^{n-1}[x_i,x_j]_C)\varphi(\alpha(x_1),\dots,\widehat{\alpha(x_i)},\dots,\widehat{\alpha(x_j)},\dots,\alpha(x_{n+2}))\\
\label{eq-37}&&+\sum_{1\leq k<i<j\leq n+1}(-1)^{i+j}(-1)^k\rho(\alpha^n(x_k)))\varphi([x_i,x_j]_C,\alpha(x_1),\dots,\widehat{\alpha(x_k)},\dots,\widehat{\alpha(x_i)},\dots,
\widehat{\alpha(x_j)},\dots,\alpha(x_{n+2}))\\
\label{eq-38}&&+\sum_{1\leq i<k<j\leq n+1}(-1)^{i+j}(-1)^{k+1}\rho(\alpha^n(x_k)))\varphi([x_i,x_j]_C,\alpha(x_1),\dots,\widehat{\alpha(x_i)},\dots,\widehat{\alpha(x_k)},\dots,
\widehat{\alpha(x_j)},\dots,\alpha(x_{n+2}))\\
\label{eq-39}&&+\sum_{1\leq i<j<k\leq n+1}(-1)^{i+j}(-1)^k\rho(\alpha^n(x_k)))\varphi([x_i,x_j]_C,\alpha(x_1),\dots,\widehat{\alpha(x_i)},\dots,\widehat{\alpha(x_j)},\dots,
\widehat{\alpha(x_k)},\dots,\alpha(x_{n+2}))\\
\label{eq-40}&&+\sum_{1\leq k<i<j\leq n+1}(-1)^{i+j}(-1)^k\mu(\alpha^n(x_{n+2})))\varphi([x_i,x_j]_C,\alpha(x_1),\dots,\widehat{\alpha(x_k)},\dots,\widehat{\alpha(x_i)},\dots,\\
\nonumber&&\widehat{\alpha(x_j)},\dots,\alpha(x_{n+1}),\alpha(x_k))\\
\label{eq-41}&&+\sum_{1\leq i<k<j\leq n+1}(-1)^{i+j}(-1)^{k+1}\mu(\alpha^n(x_{n+2})))\varphi([x_i,x_j]_C,\alpha(x_1),\dots,\widehat{\alpha(x_i)},\dots,\widehat{\alpha(x_k)},\dots,\\
\nonumber&&\widehat{\alpha(x_j)},\dots,\alpha(x_{n+1}),\alpha(x_k))\\
\label{eq-42}&&+\sum_{1\leq i<j<k\leq n+1}(-1)^{i+j}(-1)^k\mu(\alpha^n(x_{n+2})))\varphi([x_i,x_j]_C,\alpha(x_1),\dots,\widehat{\alpha(x_i)},\dots,\widehat{\alpha(x_j)},\dots,\\
\nonumber&&\widehat{\alpha(x_k)},\dots,\alpha(x_{n+1}),\alpha(x_k))\\
\label{eq-43}&&+\sum_{1\leq i<j\leq n+1}(-1)^{i+j}\mu(\alpha^n(x_{n+2}))\varphi(\alpha(x_1),\dots,\widehat{\alpha(x_i)},\dots,\widehat{\alpha(x_j)},\dots,\alpha(x_{n+1}),[x_i,x_j]_C)\\
\label{eq-44}&&-\sum_{1\leq k<i<j\leq n+1}(-1)^{i+j}(-1)^k\varphi(\alpha[x_i,x_j]_C,\alpha^2(x_1),\dots,\widehat{\alpha^2(x_k)},\dots,\widehat{\alpha^2(x_i)},\dots,\widehat{\alpha^2(x_j)},\\
\nonumber&&\dots,\alpha^2(x_{n+1}),\alpha(x_k) \cdot \alpha(x_{n+2}))\\
\label{eq-45}&&-\sum_{1\leq i<k<j\leq n+1}(-1)^{i+j}(-1)^{k+1}\varphi(\alpha[x_i,x_j]_C,\alpha^2(x_1),\dots,\widehat{\alpha^2(x_i)},\dots,\widehat{\alpha^2(x_k)},\dots,\widehat{\alpha^2(x_j)},\\
\nonumber&&\dots,\alpha^2(x_{n+1}),\alpha(x_k) \cdot \alpha(x_{n+2}))\\
\label{eq-46}&&-\sum_{1\leq i<j<k\leq n+1}(-1)^{i+j}(-1)^k\varphi(\alpha[x_i,x_j]_C,\alpha^2(x_1),\dots,\widehat{\alpha^2(x_i)},\dots,\widehat{\alpha^2(x_j)},\dots,\widehat{\alpha^2(x_k)},\\
\nonumber&&\dots,\alpha^2(x_{n+1}),\alpha(x_k) \cdot \alpha(x_{n+2}))\\
\label{eq-47}&&-\sum_{1\leq i<j\leq n+1}(-1)^{i+j}\varphi(\alpha^2(x_1),\dots,\widehat{\alpha^2(x_i)},\dots,\widehat{\alpha^2(x_j)},\dots,\alpha^2(x_{n+1}),[x_i,x_j]_C \cdot \alpha(x_{n+2}))\\
\label{eq-48}&&+\sum_{1\leq k<l<i<j\leq n+1}(-1)^{i+j}(-1)^{k+l}\varphi([\alpha(x_k),\alpha(x_l)]_C,\alpha[x_i,x_j]_C,\alpha^2(x_1),\dots,\widehat{\alpha^2(x_k)},\dots,\\
\nonumber&&\widehat{\alpha^2(x_l)},\dots,\widehat{\alpha^2(x_i)},\dots,\widehat{\alpha^2(x_j)},\dots,\alpha^2(x_{n+2}))\\
\label{eq-49}&&+\sum_{1\leq k<i<l<j\leq n+1}(-1)^{i+j}(-1)^{k+l+1}\varphi([\alpha(x_k),\alpha(x_l)]_C,\alpha[x_i,x_j]_C,\alpha^2(x_1),\dots,\widehat{\alpha^2(x_k)},\dots,\\
\nonumber&&\widehat{\alpha^2(x_i)},\dots,\widehat{\alpha^2(x_l)},\dots,\widehat{\alpha^2(x_j)},\dots,\alpha^2(x_{n+2}))\\
\label{eq-50}&&+\sum_{1\leq i<k<l<j\leq n+1}(-1)^{i+j}(-1)^{k+l}\varphi([\alpha(x_k),\alpha(x_l)]_C,\alpha[x_i,x_j]_C,\alpha^2(x_1),\dots,\widehat{\alpha^2(x_i)},\dots,\\
\nonumber&&\widehat{\alpha^2(x_k)},\dots,\widehat{\alpha^2(x_l)},\dots,\widehat{\alpha^2(x_j)},\dots,\alpha^2(x_{n+2}))\\
\label{eq-51}&&+\sum_{1\leq i<j<k<l\leq n+1}(-1)^{i+j}(-1)^{k+l}\varphi([\alpha(x_k),\alpha(x_l)]_C,\alpha[x_i,x_j]_C,\alpha^2(x_1),\dots,\widehat{\alpha^2(x_i)},\dots,\\
\nonumber&&\widehat{\alpha^2(x_j)},\dots,\widehat{\alpha^2(x_k)},\dots,\widehat{\alpha^2(x_l)},\dots,\alpha^2(x_{n+2}))\\
\label{eq-52}&&+\sum_{1\leq i<k<j<l\leq n+1}(-1)^{i+j}(-1)^{k+l+1}\varphi([\alpha(x_k),\alpha(x_l)]_C,\alpha[x_i,x_j]_C,\alpha^2(x_1),\dots,\widehat{\alpha^2(x_i)},\dots,\\
\nonumber&&\widehat{\alpha^2(x_k)},\dots,\widehat{\alpha^2(x_j)},\dots,\widehat{\alpha^2(x_l)},\dots,\alpha^2(x_{n+2}))\\
\label{eq-53}&&+\sum_{1\leq k<i<j<l\leq n+1}(-1)^{i+j}(-1)^{k+l}\varphi([\alpha(x_k),\alpha(x_l)]_C,\alpha[x_i,x_j]_C,\alpha^2(x_1),\dots,\widehat{\alpha^2(x_k)},\dots,\\
\nonumber&&\widehat{\alpha^2(x_i)},\dots,\widehat{\alpha^2(x_j)},\dots,\widehat{\alpha^2(x_l)},\dots,\alpha^2(x_{n+2}))\\
\label{eq-54}&&+\sum_{1\leq k<i<j\leq n+1}(-1)^{i+j}(-1)^k\varphi([[x_i,x_j]_C,\alpha(x_k)]_C,\alpha^2(x_1),\dots,\widehat{\alpha^2(x_k)},\dots,\widehat{\alpha^2(x_i)},\dots,
\widehat{\alpha^2(x_j)},\dots,\alpha^2(x_{n+2}))\\
\label{eq-55}&&+\sum_{1\leq i<k<j\leq n+1}(-1)^{i+j}(-1)^{k+1}\varphi([[x_i,x_j]_C,\alpha(x_k)]_C,\alpha^2(x_1),\dots,\widehat{\alpha^2(x_i)},\dots,\widehat{\alpha^2(x_k)},\dots,
\widehat{\alpha^2(x_j)},\dots,\alpha^2(x_{n+2}))\\
\label{eq-56}&&+\sum_{1\leq i<j<k\leq n+1}(-1)^{i+j}(-1)^k\varphi([[x_i,x_j]_C,\alpha(x_k)]_C,\alpha^2(x_1),\dots,\widehat{\alpha^2(x_i)},\dots,\widehat{\alpha^2(x_j)},\dots,
\widehat{\alpha^2(x_k)},\dots,\alpha^2(x_{n+2})).
\end{eqnarray}
}
The terms \eqref{eq-13} and \eqref{eq-28}, \eqref{eq-14} and \eqref{eq-27}, \eqref{eq-15} and \eqref{eq-39}, \eqref{eq-16} and \eqref{eq-38}, \eqref{eq-17} and \eqref{eq-37}, \eqref{eq-24} and \eqref{eq-42}, \eqref{eq-25} and \eqref{eq-41}, \eqref{eq-26} and \eqref{eq-40} cancel each other. By the definition of the sub-adjacent Hom-Lie algebra, the sum of \eqref{eq-54}, \eqref{eq-55} and \eqref{eq-56} is zero. By the antisymmetry condition, the term \eqref{eq-48} and \eqref{eq-51}, \eqref{eq-49} and \eqref{eq-52}, \eqref{eq-50} and \eqref{eq-53} cancel each other. By the definition of the sub-adjacent Lie bracket, the sum of \eqref{eq-22}, \eqref{eq-23} and \eqref{eq-43} is zero. By the definition of Hom-pre-Lie algebras, the sum of \eqref{eq-31}, \eqref{eq-32} and \eqref{eq-47} is zero. Since $\alpha$ is an algebra morphism, the term \eqref{eq-33} and \eqref{eq-46}, \eqref{eq-34} and \eqref{eq-45}, \eqref{eq-35} and \eqref{eq-44} cancel each other.

Since $(V,\beta,\rho,\mu)$ is a representation of the Hom-pre-Lie algebra $(A,\cdot,\alpha)$, the sum of \eqref{eq-9} and \eqref{eq-10} can be written as
{\footnotesize
\begin{equation}
\sum_{1\leq i<j\leq n+1}(-1)^{i+j+1}\rho([\alpha^{n-1}(x_i),\alpha^{n-1}(x_j)]_C)\beta\varphi(x_1,\dots,\widehat{x_i},\dots,\widehat{x_j},\dots, x_{n+2}),
\end{equation}
}
the sum of \eqref{eq-11}, \eqref{eq-19} and \eqref{eq-20} can be written as
{\footnotesize
\begin{equation}
-\sum_{1\leq j<i\leq n+1}(-1)^{i+j+1}\mu(\alpha^{n-1}(x_i) \cdot\alpha^{n-1}(x_{n+2}))\beta\varphi(x_1,\dots,\widehat{x_j},\dots,\widehat{x_i},\dots,x_{n+1}, x_j),
\end{equation}
}
and the sum of \eqref{eq-12}, \eqref{eq-18} and \eqref{eq-21} can be written as
{\footnotesize
\begin{equation}
\sum_{1\leq i<j\leq n+1}(-1)^{i+j+1}\mu(\alpha^{n-1}(x_i) \cdot\alpha^{n-1}(x_{n+2}))\beta\varphi(x_1,\dots,\widehat{x_i},\dots,\widehat{x_j},\dots,x_{n+1}, x_j).
\end{equation}
}
Thus, we have
{\footnotesize
\begin{eqnarray*}
\nonumber&&\partial_{\omega\omega}(\partial_{\omega\omega}\varphi)(x_1,\dots,x_{n+2})\\
&=&\sum_{1\leq i<j\leq n+1}(-1)^{i+j+1}\rho([\alpha^{n-1}(x_i),\alpha^{n-1}(x_j)]_C)\beta\varphi(x_1,\dots,\widehat{x_i},\dots,\widehat{x_j}, \dots,x_{n+2}),\\
&&+\sum_{1\leq i<j\leq n+1}(-1)^{i+j+1}\mu(\alpha^{n-1}(x_i) \cdot\alpha^{n-1}(x_{n+2}))\beta\varphi(x_1,\dots,\widehat{x_i},\dots,\widehat{x_j},\dots,x_{n+1}, x_j)\\
&&-\sum_{1\leq j<i\leq n+1}(-1)^{i+j+1}\mu(\alpha^{n-1}(x_i) \cdot\alpha^{n-1}(x_{n+2}))\beta\varphi(x_1,\dots,\widehat{x_j},\dots,\widehat{x_i},\dots,x_{n+1}, x_j)\\
&&-\sum_{1\leq i<j\leq n+1}(-1)^{i+j+1}\rho(\alpha^{n-1}[x_i,x_j]_C)\varphi(\alpha(x_1),\dots,\widehat{\alpha(x_i)},\dots,\widehat{\alpha(x_j)},\dots,\alpha(x_{n+2}))\\
&&-\sum_{1\leq i<j\leq n+1}(-1)^{i+j+1}\mu(\alpha^{n-1}(x_i \cdot x_{n+2}))\varphi(\alpha(x_1),\dots,\widehat{\alpha(x_i)},\dots,\widehat{\alpha(x_j)},\dots,\alpha(x_{n+1}),\alpha(x_j))\\
&&+\sum_{1\leq j<i\leq n+1}(-1)^{i+j+1}\mu(\alpha^{n-1}(x_i \cdot x_{n+2}))\varphi(\alpha(x_1),\dots,\widehat{\alpha(x_j)},\dots,\widehat{\alpha(x_i)},\dots,\alpha(x_{n+1}),\alpha(x_j)).
\end{eqnarray*}
}
For all $x_1,\dots,x_{n+2}\in A$, we have
{\footnotesize
\begin{eqnarray*}
&&\partial_{\alpha\omega}(\partial_{\omega\alpha}\varphi)(x_1,\dots,x_{n+2})\\
&=&\sum_{1\leq i<j\leq n+1}(-1)^{i+j}\rho([\alpha^{n-1}(x_i),\alpha^{n-1}(x_j)]_C) (\partial_{\omega\alpha}\varphi)(x_1,\dots,\hat{x_i},\dots,\hat{x_j},\dots,x_{n+2})\\
&&+\sum_{1\leq i<j\leq n+1}(-1)^{i+j}\mu(\alpha^{n-1}(x_i)\cdot\alpha^{n-1}(x_{n+2}))(\partial_{\omega\alpha}\varphi)(x_1,\dots,\hat{x_i},\dots,\hat{x_j},\dots,x_{n+1},x_j)\\
&&-\sum_{1\leq i<j\leq n+1}(-1)^{i+j}\mu(\alpha^{n-1}(x_j)\cdot\alpha^{n-1}(x_{n+2}))(\partial_{\omega\alpha}\varphi)(x_1,\dots,\hat{x_i},\dots,\hat{x_j},\dots,x_{n+1},x_i)\\
&=&\sum_{1\leq i<j\leq n+1}(-1)^{i+j}\rho([\alpha^{n-1}(x_i),\alpha^{n-1}(x_j)]_C)\beta\varphi(x_1,\dots,\widehat{x_i},\dots,\widehat{x_j}, \dots,x_{n+2}),\\
&&-\sum_{1\leq i<j\leq n+1}(-1)^{i+j}\rho([\alpha^{n-1}(x_i),\alpha^{n-1}(x_j)]_C)\varphi(\alpha(x_1),\dots,\widehat{\alpha(x_i)},\dots,\widehat{\alpha(x_j)},\dots,\alpha(x_{n+2}))\\
&&+\sum_{1\leq i<j\leq n+1}(-1)^{i+j}\mu(\alpha^{n-1}(x_i) \cdot\alpha^{n-1}(x_{n+2}))\beta\varphi(x_1,\dots,\widehat{x_i},\dots,\widehat{x_j},\dots,x_{n+1}, x_j)\\
&&-\sum_{1\leq i<j\leq n+1}(-1)^{i+j}\mu(\alpha^{n-1}(x_i) \cdot\alpha^{n-1}(x_{n+2}))\varphi(\alpha(x_1),\dots,\widehat{\alpha(x_i)},\dots,\widehat{\alpha(x_j)},\dots,\alpha(x_{n+1}),\alpha(x_j))\\
&&-\sum_{1\leq i<j\leq n+1}(-1)^{i+j}\mu(\alpha^{n-1}(x_j) \cdot\alpha^{n-1}(x_{n+2}))\beta\varphi(x_1,\dots,\widehat{x_i},\dots,\widehat{x_j},\dots,x_{n+1}, x_i)\\
&&+\sum_{1\leq i<j\leq n+1}(-1)^{i+j}\mu(\alpha^{n-1}(x_j) \cdot\alpha^{n-1}(x_{n+2}))\varphi(\alpha(x_1),\dots,\widehat{\alpha(x_i)},\dots,\widehat{\alpha(x_j)},\dots,\alpha(x_{n+1}),\alpha(x_i)).\\
\end{eqnarray*}
}
Since $\alpha$ is an algebra morphism, we obtain that $\partial_{\omega\omega}\circ \partial_{\omega\omega}+\partial_{\alpha\omega}\circ \partial_{\omega\alpha}$=0.

For all $x_1,\dots,x_{n+2}\in A$, we have
{\footnotesize
\begin{eqnarray}
\nonumber&&\partial_{\omega\omega}(\partial_{\alpha\omega}\psi)(x_1,\dots,x_{n+2})\\
\nonumber&=&\sum_{i=1}^{n+1}(-1)^{i+1}\rho(\alpha^n(x_i))(\partial_{\alpha\omega}\psi)(x_1,\dots,\widehat{x_i},\dots,x_{n+2})\\
\nonumber&&+\sum_{i=1}^{n+1}(-1)^{i+1}\mu(\alpha^n(x_{n+2}))(\partial_{\alpha\omega}\psi)(x_1,\dots,\widehat{x_i},\dots,x_{n+1},x_i)\\
\nonumber&&-\sum_{i=1}^{n+1}(-1)^{i+1}(\partial_{\alpha\omega}\psi)(\alpha(x_1),\dots,\widehat{\alpha(x_i)},\dots,\alpha(x_{n+1}),x_i\cdot x_{n+2})\\
\nonumber&&+\sum_{1\leq i<j\leq n+1}(-1)^{i+j} (\partial_{\alpha\omega}\psi)([x_i,x_j]_C,\alpha(x_1),\dots,\widehat{\alpha(x_i)},\dots,\widehat{\alpha(x_j)},\dots,\alpha(x_{n+2}))\\
\label{eq-60}&=&\sum_{1\leq j<k<i\leq n+1}(-1)^{i+1}(-1)^{j+k}\rho(\alpha^n(x_i))\rho([\alpha^{n-2}(x_j),\alpha^{n-2}(x_k)]_C)\psi(x_1,\dots,\hat{x_j},\dots,\hat{x_k}
\dots,\hat{x_i},\dots,x_{n+2})\\
\label{eq-61}&&+\sum_{1\leq j<i<k\leq n+1}(-1)^{i+1}(-1)^{j+k-1}\rho(\alpha^n(x_i))\rho([\alpha^{n-2}(x_j),\alpha^{n-2}(x_k)]_C)\psi(x_1,\dots,\hat{x_j},\dots,\hat{x_i}
\dots,\hat{x_k},\dots,x_{n+2})\\
\label{eq-62}&&+\sum_{1\leq i<j<k\leq n+1}(-1)^{i+1}(-1)^{j+k}\rho(\alpha^n(x_i))\rho([\alpha^{n-2}(x_j),\alpha^{n-2}(x_k)]_C)\psi(x_1,\dots,\hat{x_i},\dots,\hat{x_j}
\dots,\hat{x_k},\dots,x_{n+2})\\
\label{eq-63}&&+\sum_{1\leq j<k<i\leq n+1}(-1)^{i+1}(-1)^{j+k}\rho(\alpha^n(x_i))\mu(\alpha^{n-2}(x_j)\cdot\alpha^{n-2}(x_{n+2}))\psi(x_1,\dots,\hat{x_j},\dots,\hat{x_k}
\dots,\hat{x_i},\dots,x_{n+1},x_k)\\
\label{eq-64}&&+\sum_{1\leq j<i<k\leq n+1}(-1)^{i+1}(-1)^{j+k-1}\rho(\alpha^n(x_i))\mu(\alpha^{n-2}(x_j)\cdot\alpha^{n-2}(x_{n+2}))\psi(x_1,\dots,\hat{x_j},\dots,\hat{x_i}
\dots,\hat{x_k},\dots,x_{n+1},x_k)\\
\label{eq-65}&&+\sum_{1\leq i<j<k\leq n+1}(-1)^{i+1}(-1)^{j+k}\rho(\alpha^n(x_i))\mu(\alpha^{n-2}(x_j)\cdot\alpha^{n-2}(x_{n+2}))\psi(x_1,\dots,\hat{x_i},\dots,\hat{x_j}
\dots,\hat{x_k},\dots,x_{n+1},x_k)\\
\label{eq-66}&&-\sum_{1\leq j<k<i\leq n+1}(-1)^{i+1}(-1)^{j+k}\rho(\alpha^n(x_i))\mu(\alpha^{n-2}(x_k)\cdot\alpha^{n-2}(x_{n+2}))\psi(x_1,\dots,\hat{x_j},\dots,\hat{x_k}
\dots,\hat{x_i},\dots,x_{n+1},x_j)\\
\label{eq-67}&&-\sum_{1\leq j<i<k\leq n+1}(-1)^{i+1}(-1)^{j+k-1}\rho(\alpha^n(x_i))\mu(\alpha^{n-2}(x_k)\cdot\alpha^{n-2}(x_{n+2}))\psi(x_1,\dots,\hat{x_j},\dots,\hat{x_i}
\dots,\hat{x_k},\dots,x_{n+1},x_j)\\
\label{eq-68}&&-\sum_{1\leq i<j<k\leq n+1}(-1)^{i+1}(-1)^{j+k}\rho(\alpha^n(x_i))\mu(\alpha^{n-2}(x_k)\cdot\alpha^{n-2}(x_{n+2}))\psi(x_1,\dots,\hat{x_i},\dots,\hat{x_j}
\dots,\hat{x_k},\dots,x_{n+1},x_j)\\
\label{eq-69}&&+\sum_{1\leq j<k<i\leq n+1}(-1)^{i+1}(-1)^{j+k}\mu(\alpha^n(x_{n+2}))\rho([\alpha^{n-2}(x_j),\alpha^{n-2}(x_k)]_C)\psi(x_1,\dots,\hat{x_j},\dots,\hat{x_k}
\dots,\hat{x_i},\dots,x_{n+1},x_i)\\
\label{eq-70}&&+\sum_{1\leq j<i<k\leq n+1}(-1)^{i+1}(-1)^{j+k-1}\mu(\alpha^n(x_{n+2}))\rho([\alpha^{n-2}(x_j),\alpha^{n-2}(x_k)]_C)\psi(x_1,\dots,\hat{x_j},\dots,\hat{x_i}
\dots,\hat{x_k},\dots,x_{n+1},x_i)\\
\label{eq-71}&&+\sum_{1\leq i<j<k\leq n+1}(-1)^{i+1}(-1)^{j+k}\mu(\alpha^n(x_{n+2}))\rho([\alpha^{n-2}(x_j),\alpha^{n-2}(x_k)]_C)\psi(x_1,\dots,\hat{x_i},\dots,\hat{x_j}
\dots,\hat{x_k},\dots,x_{n+1},x_i)\\
\label{eq-72}&&+\sum_{1\leq j<k<i\leq n+1}(-1)^{i+1}(-1)^{j+k}\mu(\alpha^n(x_{n+2}))\mu(\alpha^{n-2}(x_j)\cdot\alpha^{n-2}(x_i))\psi(x_1,\dots,\hat{x_j},\dots,\hat{x_k}
\dots,\hat{x_i},\dots,x_{n+1},x_k)\\
\label{eq-73}&&+\sum_{1\leq j<i<k\leq n+1}(-1)^{i+1}(-1)^{j+k-1}\mu(\alpha^n(x_{n+2}))\mu(\alpha^{n-2}(x_j)\cdot\alpha^{n-2}(x_i))\psi(x_1,\dots,\hat{x_j},\dots,\hat{x_i}
\dots,\hat{x_k},\dots,x_{n+1},x_k)\\
\label{eq-74}&&+\sum_{1\leq i<j<k\leq n+1}(-1)^{i+1}(-1)^{j+k}\mu(\alpha^n(x_{n+2}))\mu(\alpha^{n-2}(x_j)\cdot\alpha^{n-2}(x_i))\psi(x_1,\dots,\hat{x_i},\dots,\hat{x_j}
\dots,\hat{x_k},\dots,x_{n+1},x_k)\\
\label{eq-75}&&-\sum_{1\leq j<k<i\leq n+1}(-1)^{i+1}(-1)^{j+k}\mu(\alpha^n(x_{n+2}))\mu(\alpha^{n-2}(x_k)\cdot\alpha^{n-2}(x_i))\psi(x_1,\dots,\hat{x_j},\dots,\hat{x_k}
\dots,\hat{x_i},\dots,x_{n+1},x_j)\\
\label{eq-76}&&-\sum_{1\leq j<i<k\leq n+1}(-1)^{i+1}(-1)^{j+k-1}\mu(\alpha^n(x_{n+2}))\mu(\alpha^{n-2}(x_k)\cdot\alpha^{n-2}(x_i))\psi(x_1,\dots,\hat{x_j},\dots,\hat{x_i}
\dots,\hat{x_k},\dots,x_{n+1},x_j)\\
\label{eq-77}&&-\sum_{1\leq i<j<k\leq n+1}(-1)^{i+1}(-1)^{j+k}\mu(\alpha^n(x_{n+2}))\mu(\alpha^{n-2}(x_k)\cdot\alpha^{n-2}(x_i))\psi(x_1,\dots,\hat{x_i},\dots,\hat{x_j}
\dots,\hat{x_k},\dots,x_{n+1},x_j)\\
\label{eq-78}&&-\sum_{1\leq j<k<i\leq n+1}(-1)^{i+1}(-1)^{j+k}\rho([\alpha^{n-1}(x_j),\alpha^{n-1}(x_k)]_C)\psi(\alpha(x_1),\dots,\widehat{\alpha(x_j)},\dots,\widehat{\alpha(x_k)},\\
\nonumber&&\dots,\widehat{\alpha(x_i)},\dots,\alpha(x_{n+1}),x_i\cdot x_{n+2})\\
\label{eq-79}&&-\sum_{1\leq j<i<k\leq n+1}(-1)^{i+1}(-1)^{j+k-1}\rho([\alpha^{n-1}(x_j),\alpha^{n-1}(x_k)]_C)\psi(\alpha(x_1),\dots,\widehat{\alpha(x_j)},\dots,\widehat{\alpha(x_i)},\\
\nonumber&&\dots,\widehat{\alpha(x_k)},\dots,\alpha(x_{n+1}),x_i\cdot x_{n+2})\\
\label{eq-80}&&-\sum_{1\leq i<j<k\leq n+1}(-1)^{i+1}(-1)^{j+k}\rho([\alpha^{n-1}(x_j),\alpha^{n-1}(x_k)]_C)\psi(\alpha(x_1),\dots,\widehat{\alpha(x_i)},\dots,\widehat{\alpha(x_j)},\\
\nonumber&&\dots,\widehat{\alpha(x_k)},\dots,\alpha(x_{n+1}),x_i\cdot x_{n+2})\\
\label{eq-81}&&-\sum_{1\leq j<k<i\leq n+1}(-1)^{i+1}(-1)^{j+k}\mu(\alpha^{n-1}(x_j)\cdot\alpha^{n-2}(x_i\cdot x_{n+2}))\psi(\alpha(x_1),\dots,\widehat{\alpha(x_j)},\dots,\widehat{\alpha(x_k)},\\
\nonumber&&\dots,\widehat{\alpha(x_i)},\dots,\alpha(x_{n+1}),\alpha(x_k))\\
\label{eq-82}&&-\sum_{1\leq j<i<k\leq n+1}(-1)^{i+1}(-1)^{j+k-1}\mu(\alpha^{n-1}(x_j)\cdot\alpha^{n-2}(x_i\cdot x_{n+2}))\psi(\alpha(x_1),\dots,\widehat{\alpha(x_j)},\dots,\widehat{\alpha(x_i)},\\
\nonumber&&\dots,\widehat{\alpha(x_k)},\dots,\alpha(x_{n+1}),\alpha(x_k))\\
\label{eq-83}&&-\sum_{1\leq i<j<k\leq n+1}(-1)^{i+1}(-1)^{j+k}\mu(\alpha^{n-1}(x_j)\cdot\alpha^{n-2}(x_i\cdot x_{n+2}))\psi(\alpha(x_1),\dots,\widehat{\alpha(x_i)},\dots,\widehat{\alpha(x_j)},\\
\nonumber&&\dots,\widehat{\alpha(x_k)},\dots,\alpha(x_{n+1}),\alpha(x_k))\\
\label{eq-84}&&+\sum_{1\leq j<k<i\leq n+1}(-1)^{i+1}(-1)^{j+k}\mu(\alpha^{n-1}(x_k)\cdot\alpha^{n-2}(x_i\cdot x_{n+2}))\psi(\alpha(x_1),\dots,\widehat{\alpha(x_j)},\dots,\widehat{\alpha(x_k)},\\
\nonumber&&\dots,\widehat{\alpha(x_i)},\dots,\alpha(x_{n+1}),\alpha(x_j))\\
\label{eq-85}&&+\sum_{1\leq j<i<k\leq n+1}(-1)^{i+1}(-1)^{j+k-1}\mu(\alpha^{n-1}(x_k)\cdot\alpha^{n-2}(x_i\cdot x_{n+2}))\psi(\alpha(x_1),\dots,\widehat{\alpha(x_j)},\dots,\widehat{\alpha(x_i)},\\
\nonumber&&\dots,\widehat{\alpha(x_k)},\dots,\alpha(x_{n+1}),\alpha(x_j))\\
\label{eq-86}&&+\sum_{1\leq i<j<k\leq n+1}(-1)^{i+1}(-1)^{j+k}\mu(\alpha^{n-1}(x_k)\cdot\alpha^{n-2}(x_i\cdot x_{n+2}))\psi(\alpha(x_1),\dots,\widehat{\alpha(x_i)},\dots,\widehat{\alpha(x_j)},\\
\nonumber&&\dots,\widehat{\alpha(x_k)},\dots,\alpha(x_{n+1}),\alpha(x_j))\\
\label{eq-87}&&+\sum_{1\leq k<l<i<j\leq n+1}(-1)^{i+j}(-1)^{k+l}\rho([\alpha^{n-1}(x_k),\alpha^{n-1}(x_l)]_C)\psi([x_i,x_j]_C,\alpha(x_1),\dots,\widehat{\alpha(x_k)},\dots,\\
\nonumber&&\widehat{\alpha(x_l)},\dots,\widehat{\alpha(x_i)},\dots,\widehat{\alpha(x_j)},\dots,\alpha(x_{n+2}))\\
\label{eq-88}&&+\sum_{1\leq k<i<l<j\leq n+1}(-1)^{i+j}(-1)^{k+l-1}\rho([\alpha^{n-1}(x_k),\alpha^{n-1}(x_l)]_C)\psi([x_i,x_j]_C,\alpha(x_1),\dots,\widehat{\alpha(x_k)},\dots,\\
\nonumber&&\widehat{\alpha(x_i)},\dots,\widehat{\alpha(x_l)},\dots,\widehat{\alpha(x_j)},\dots,\alpha(x_{n+2}))\\
\label{eq-89}&&+\sum_{1\leq k<i<j<l\leq n+1}(-1)^{i+j}(-1)^{k+l}\rho([\alpha^{n-1}(x_k),\alpha^{n-1}(x_l)]_C)\psi([x_i,x_j]_C,\alpha(x_1),\dots,\widehat{\alpha(x_k)},\dots,\\
\nonumber&&\widehat{\alpha(x_i)},\dots,\widehat{\alpha(x_j)},\dots,\widehat{\alpha(x_l)},\dots,\alpha(x_{n+2}))\\
\label{eq-90}&&+\sum_{1\leq i<k<l<j\leq n+1}(-1)^{i+j}(-1)^{k+l}\rho([\alpha^{n-1}(x_k),\alpha^{n-1}(x_l)]_C)\psi([x_i,x_j]_C,\alpha(x_1),\dots,\widehat{\alpha(x_i)},\dots,\\
\nonumber&&\widehat{\alpha(x_k)},\dots,\widehat{\alpha(x_l)},\dots,\widehat{\alpha(x_j)},\dots,\alpha(x_{n+2}))\\
\label{eq-91}&&+\sum_{1\leq i<k<j<l\leq n+1}(-1)^{i+j}(-1)^{k+l-1}\rho([\alpha^{n-1}(x_k),\alpha^{n-1}(x_l)]_C)\psi([x_i,x_j]_C,\alpha(x_1),\dots,\widehat{\alpha(x_i)},\dots,\\
\nonumber&&\widehat{\alpha(x_k)},\dots,\widehat{\alpha(x_j)},\dots,\widehat{\alpha(x_l)},\dots,\alpha(x_{n+2}))\\
\label{eq-92}&&+\sum_{1\leq i<j<k<l\leq n+1}(-1)^{i+j}(-1)^{k+l}\rho([\alpha^{n-1}(x_k),\alpha^{n-1}(x_l)]_C)\psi([x_i,x_j]_C,\alpha(x_1),\dots,\widehat{\alpha(x_i)},\dots,\\
\nonumber&&\widehat{\alpha(x_j)},\dots,\widehat{\alpha(x_k)},\dots,\widehat{\alpha(x_l)},\dots,\alpha(x_{n+2}))\\
\label{eq-93}&&+\sum_{1\leq k<i<j\leq n+1}(-1)^{i+j}(-1)^k\rho([\alpha^{n-2}[x_i,x_j]_C,\alpha^{n-1}(x_k)]_C)\psi(\alpha(x_1),\dots,\widehat{\alpha(x_k)},\dots,\widehat{\alpha(x_i)},\\
\nonumber&&\dots,\widehat{\alpha(x_j)},\dots,\alpha(x_{n+2}))\\
\label{eq-94}&&+\sum_{1\leq i<k<j\leq n+1}(-1)^{i+j}(-1)^{k+1}\rho([\alpha^{n-2}[x_i,x_j]_C,\alpha^{n-1}(x_k)]_C)\psi(\alpha(x_1),\dots,\widehat{\alpha(x_i)},\dots,\widehat{\alpha(x_k)},\\
\nonumber&&\dots,\widehat{\alpha(x_j)},\dots,\alpha(x_{n+2}))\\
\label{eq-95}&&+\sum_{1\leq i<j<k\leq n+1}(-1)^{i+j}(-1)^k\rho([\alpha^{n-2}[x_i,x_j]_C,\alpha^{n-1}(x_k)]_C)\psi(\alpha(x_1),\dots,\widehat{\alpha(x_i)},\dots,\widehat{\alpha(x_j)},\\
\nonumber&&\dots,\widehat{\alpha(x_k)},\dots,\alpha(x_{n+2}))\\
\label{eq-96}&&+\sum_{1\leq k<l<i<j\leq n+1}(-1)^{i+j}(-1)^{k+l}\mu(\alpha^{n-1}(x_k)\cdot\alpha^{n-1}(x_{n+2}))\psi([x_i,x_j]_C,\alpha(x_1),\dots,\widehat{\alpha(x_k)},\dots,\\
\nonumber&&\widehat{\alpha(x_l)},\dots,\widehat{\alpha(x_i)},\dots,\widehat{\alpha(x_j)},\dots,\alpha(x_{n+1}),\alpha(x_l))\\
\label{eq-97}&&+\sum_{1\leq k<i<l<j\leq n+1}(-1)^{i+j}(-1)^{k+l-1}\mu(\alpha^{n-1}(x_k)\cdot\alpha^{n-1}(x_{n+2}))\psi([x_i,x_j]_C,\alpha(x_1),\dots,\widehat{\alpha(x_k)},\dots,\\
\nonumber&&\widehat{\alpha(x_i)},\dots,\widehat{\alpha(x_l)},\dots,\widehat{\alpha(x_j)},\dots,\alpha(x_{n+1}),\alpha(x_l))\\
\label{eq-98}&&+\sum_{1\leq k<i<j<l\leq n+1}(-1)^{i+j}(-1)^{k+l}\mu(\alpha^{n-1}(x_k)\cdot\alpha^{n-1}(x_{n+2}))\psi([x_i,x_j]_C,\alpha(x_1),\dots,\widehat{\alpha(x_k)},\dots,\\
\nonumber&&\widehat{\alpha(x_i)},\dots,\widehat{\alpha(x_j)},\dots,\widehat{\alpha(x_l)},\dots,\alpha(x_{n+1}),\alpha(x_l))\\
\label{eq-99}&&+\sum_{1\leq i<k<l<j\leq n+1}(-1)^{i+j}(-1)^{k+l}\mu(\alpha^{n-1}(x_k)\cdot\alpha^{n-1}(x_{n+2}))\psi([x_i,x_j]_C,\alpha(x_1),\dots,\widehat{\alpha(x_i)},\dots,\\
\nonumber&&\widehat{\alpha(x_k)},\dots,\widehat{\alpha(x_l)},\dots,\widehat{\alpha(x_j)},\dots,\alpha(x_{n+1}),\alpha(x_l))\\
\label{eq-100}&&+\sum_{1\leq i<k<j<l\leq n+1}(-1)^{i+j}(-1)^{k+l-1}\mu(\alpha^{n-1}(x_k)\cdot\alpha^{n-1}(x_{n+2}))\psi([x_i,x_j]_C,\alpha(x_1),\dots,\widehat{\alpha(x_i)},\dots,\\
\nonumber&&\widehat{\alpha(x_k)},\dots,\widehat{\alpha(x_j)},\dots,\widehat{\alpha(x_l)},\dots,\alpha(x_{n+1}),\alpha(x_l))\\
\label{eq-101}&&+\sum_{1\leq i<j<k<l\leq n+1}(-1)^{i+j}(-1)^{k+l}\mu(\alpha^{n-1}(x_k)\cdot\alpha^{n-1}(x_{n+2}))\psi([x_i,x_j]_C,\alpha(x_1),\dots,\widehat{\alpha(x_i)},\dots,\\
\nonumber&&\widehat{\alpha(x_j)},\dots,\widehat{\alpha(x_k)},\dots,\widehat{\alpha(x_l)},\dots,\alpha(x_{n+1}),\alpha(x_l))\\
\label{eq-102}&&-\sum_{1\leq k<l<i<j\leq n+1}(-1)^{i+j}(-1)^{k+l}\mu(\alpha^{n-1}(x_l)\cdot\alpha^{n-1}(x_{n+2}))\psi([x_i,x_j]_C,\alpha(x_1),\dots,\widehat{\alpha(x_k)},\dots,\\
\nonumber&&\widehat{\alpha(x_l)},\dots,\widehat{\alpha(x_i)},\dots,\widehat{\alpha(x_j)},\dots,\alpha(x_{n+1}),\alpha(x_k))\\
\label{eq-103}&&-\sum_{1\leq k<i<l<j\leq n+1}(-1)^{i+j}(-1)^{k+l-1}\mu(\alpha^{n-1}(x_l)\cdot\alpha^{n-1}(x_{n+2}))\psi([x_i,x_j]_C,\alpha(x_1),\dots,\widehat{\alpha(x_k)},\dots,\\
\nonumber&&\widehat{\alpha(x_i)},\dots,\widehat{\alpha(x_l)},\dots,\widehat{\alpha(x_j)},\dots,\alpha(x_{n+1}),\alpha(x_k))\\
\label{eq-104}&&-\sum_{1\leq k<i<j<l\leq n+1}(-1)^{i+j}(-1)^{k+l}\mu(\alpha^{n-1}(x_l)\cdot\alpha^{n-1}(x_{n+2}))\psi([x_i,x_j]_C,\alpha(x_1),\dots,\widehat{\alpha(x_k)},\dots,\\
\nonumber&&\widehat{\alpha(x_i)},\dots,\widehat{\alpha(x_j)},\dots,\widehat{\alpha(x_l)},\dots,\alpha(x_{n+1}),\alpha(x_k))\\
\label{eq-105}&&-\sum_{1\leq i<k<l<j\leq n+1}(-1)^{i+j}(-1)^{k+l}\mu(\alpha^{n-1}(x_l)\cdot\alpha^{n-1}(x_{n+2}))\psi([x_i,x_j]_C,\alpha(x_1),\dots,\widehat{\alpha(x_i)},\dots,\\
\nonumber&&\widehat{\alpha(x_k)},\dots,\widehat{\alpha(x_l)},\dots,\widehat{\alpha(x_j)},\dots,\alpha(x_{n+1}),\alpha(x_k))\\
\label{eq-106}&&-\sum_{1\leq i<k<j<l\leq n+1}(-1)^{i+j}(-1)^{k+l-1}\mu(\alpha^{n-1}(x_l)\cdot\alpha^{n-1}(x_{n+2}))\psi([x_i,x_j]_C,\alpha(x_1),\dots,\widehat{\alpha(x_i)},\dots,\\
\nonumber&&\widehat{\alpha(x_k)},\dots,\widehat{\alpha(x_j)},\dots,\widehat{\alpha(x_l)},\dots,\alpha(x_{n+1}),\alpha(x_k))\\
\label{eq-107}&&-\sum_{1\leq i<j<k<l\leq n+1}(-1)^{i+j}(-1)^{k+l}\mu(\alpha^{n-1}(x_l)\cdot\alpha^{n-1}(x_{n+2}))\psi([x_i,x_j]_C,\alpha(x_1),\dots,\widehat{\alpha(x_i)},\dots,\\
\nonumber&&\widehat{\alpha(x_j)},\dots,\widehat{\alpha(x_k)},\dots,\widehat{\alpha(x_l)},\dots,\alpha(x_{n+1}),\alpha(x_k))\\
\label{eq-108}&&+\sum_{1\leq k<i<j\leq n+1}(-1)^{i+j}(-1)^k\mu(\alpha^{n-2}[x_i,x_j]_C\cdot\alpha^{n-1}(x_{n+2}))\psi(\alpha(x_1),\dots,\widehat{\alpha(x_k)},\dots,\widehat{\alpha(x_i)},\\
\nonumber&&\dots,\widehat{\alpha(x_j)},\dots,\alpha(x_{n+1}),\alpha(x_k))\\
\label{eq-109}&&+\sum_{1\leq i<k<j\leq n+1}(-1)^{i+j}(-1)^{k+1}\mu(\alpha^{n-2}[x_i,x_j]_C\cdot\alpha^{n-1}(x_{n+2}))\psi(\alpha(x_1),\dots,\widehat{\alpha(x_i)},\dots,\widehat{\alpha(x_k)},\\
\nonumber&&\dots,\widehat{\alpha(x_j)},\dots,\alpha(x_{n+1}),\alpha(x_k))\\
\label{eq-110}&&+\sum_{1\leq i<j<k\leq n+1}(-1)^{i+j}(-1)^k\mu(\alpha^{n-2}[x_i,x_j]_C\cdot\alpha^{n-1}(x_{n+2}))\psi(\alpha(x_1),\dots,\widehat{\alpha(x_i)},\dots,\widehat{\alpha(x_j)},\\
\nonumber&&\dots,\widehat{\alpha(x_k)},\dots,\alpha(x_{n+1}),\alpha(x_k))\\
\label{eq-111}&&-\sum_{1\leq k<i<j\leq n+1}(-1)^{i+j}(-1)^k\mu(\alpha^{n-1}(x_k)\cdot\alpha^{n-1}(x_{n+2}))\psi(\alpha(x_1),\dots,\widehat{\alpha(x_k)},\dots,\widehat{\alpha(x_i)},\\
\nonumber&&\dots,\widehat{\alpha(x_j)},\dots,\alpha(x_{n+1}),[x_i,x_j]_C)\\
\label{eq-112}&&-\sum_{1\leq i<k<j\leq n+1}(-1)^{i+j}(-1)^{k+1}\mu(\alpha^{n-1}(x_k)\cdot\alpha^{n-1}(x_{n+2}))\psi(\alpha(x_1),\dots,\widehat{\alpha(x_i)},\dots,\widehat{\alpha(x_k)},\\
\nonumber&&\dots,\widehat{\alpha(x_j)},\dots,\alpha(x_{n+1}),[x_i,x_j]_C)\\
\label{eq-113}&&-\sum_{1\leq i<j<k\leq n+1}(-1)^{i+j}(-1)^k\mu(\alpha^{n-1}(x_k)\cdot\alpha^{n-1}(x_{n+2}))\psi(\alpha(x_1),\dots,\widehat{\alpha(x_i)},\dots,\widehat{\alpha(x_j)},\\
\nonumber&&\dots,\widehat{\alpha(x_k)},\dots,\alpha(x_{n+1}),[x_i,x_j]_C),
\end{eqnarray}
}
and
{\footnotesize
\begin{eqnarray}
\nonumber&&\partial_{\alpha\omega}(\partial_{\alpha\alpha} \psi)(x_1,\dots,x_{n+2})\\
\nonumber&=&\sum_{1\leq i<j\leq n+1}(-1)^{i+j}\rho([\alpha^{n-1}(x_i),\alpha^{n-1}(x_j)]_C)(\partial_{\alpha\alpha} \psi)(x_1,\dots,\hat{x_i},\dots,\hat{x_j},\dots,x_{n+2})\\
\nonumber&&+\sum_{1\leq i<j\leq n+1}(-1)^{i+j}\mu(\alpha^{n-1}(x_i)\cdot\alpha^{n-1}(x_{n+2}))(\partial_{\alpha\alpha} \psi)(x_1,\dots,\hat{x_i},\dots,\hat{x_j},\dots,x_{n+1},x_j)\\
\nonumber&&-\sum_{1\leq i<j\leq n+1}(-1)^{i+j}\mu(\alpha^{n-1}(x_j)\cdot\alpha^{n-1}(x_{n+2}))(\partial_{\alpha\alpha} \psi)(x_1,\dots,\hat{x_i},\dots,\hat{x_j},\dots,x_{n+1},x_i)\\
\label{eq-114}&=&\sum_{1\leq k<i<j\leq n+1}(-1)^{i+j}(-1)^k\rho([\alpha^{n-1}(x_i),\alpha^{n-1}(x_j)]_C)\rho(\alpha^{n-1}(x_k))\psi(x_1,\dots,\hat{x_k},\dots,\hat{x_i},
\dots,\hat{x_j}\dots,x_{n+2})\\
\label{eq-115}&&+\sum_{1\leq i<k<j\leq n+1}(-1)^{i+j}(-1)^{k+1}\rho([\alpha^{n-1}(x_i),\alpha^{n-1}(x_j)]_C)\rho(\alpha^{n-1}(x_k))\psi(x_1,\dots,\hat{x_i},\dots,\hat{x_k},
\dots,\hat{x_j}\dots,x_{n+2})\\
\label{eq-116}&&+\sum_{1\leq i<j<k\leq n+1}(-1)^{i+j}(-1)^k\rho([\alpha^{n-1}(x_i),\alpha^{n-1}(x_j)]_C)\rho(\alpha^{n-1}(x_k))\psi(x_1,\dots,\hat{x_i},\dots,\hat{x_j},
\dots,\hat{x_k},\dots,x_{n+2})\\
\label{eq-117}&&+\sum_{1\leq k<i<j\leq n+1}(-1)^{i+j}(-1)^k\rho([\alpha^{n-1}(x_i),\alpha^{n-1}(x_j)]_C)\mu(\alpha^{n-1}(x_{n+2}))\psi(x_1,\dots,\hat{x_k},\dots,\hat{x_i},
\dots,\hat{x_j}\dots,x_{n+1},x_k)\\
\label{eq-118}&&+\sum_{1\leq i<k<j\leq n+1}(-1)^{i+j}(-1)^{k+1}\rho([\alpha^{n-1}(x_i),\alpha^{n-1}(x_j)]_C)\mu(\alpha^{n-1}(x_{n+2}))\psi(x_1,\dots,\hat{x_i},\dots,\hat{x_k},
\dots,\hat{x_j}\dots,x_{n+1},x_k)\\
\label{eq-119}&&+\sum_{1\leq i<j<k\leq n+1}(-1)^{i+j}(-1)^k\rho([\alpha^{n-1}(x_i),\alpha^{n-1}(x_j)]_C)\mu(\alpha^{n-1}(x_{n+2}))\psi(x_1,\dots,\hat{x_i},\dots,\hat{x_j},
\dots,\hat{x_k}\dots,x_{n+1},x_k)\\
\label{eq-120}&&-\sum_{1\leq k<i<j\leq n+1}(-1)^{i+j}(-1)^k\rho([\alpha^{n-1}(x_i),\alpha^{n-1}(x_j)]_C)\psi(\alpha(x_1),\dots,\widehat{\alpha(x_k)},\dots,\widehat{\alpha(x_i)},\\
\nonumber&&\dots,\widehat{\alpha(x_j)},\dots,\alpha(x_{n+1}),x_k\cdot x_{n+2})\\
\label{eq-121}&&-\sum_{1\leq i<k<j\leq n+1}(-1)^{i+j}(-1)^{k+1}\rho([\alpha^{n-1}(x_i),\alpha^{n-1}(x_j)]_C)\psi(\alpha(x_1),\dots,\widehat{\alpha(x_i)},\dots,\widehat{\alpha(x_k)},\\
\nonumber&&\dots,\widehat{\alpha(x_j)},\dots,\alpha(x_{n+1}),x_k\cdot x_{n+2})\\
\label{eq-122}&&-\sum_{1\leq i<j<k\leq n+1}(-1)^{i+j}(-1)^k\rho([\alpha^{n-1}(x_i),\alpha^{n-1}(x_j)]_C)\psi(\alpha(x_1),\dots,\widehat{\alpha(x_i)},\dots,\widehat{\alpha(x_j)},\\
\nonumber&&\dots,\widehat{\alpha(x_k)},\dots,\alpha(x_{n+1}),x_k\cdot x_{n+2})\\
\label{eq-123}&&+\sum_{1\leq k<l<i<j\leq n+1}(-1)^{i+j}(-1)^{k+l-1}\rho([\alpha^{n-1}(x_i),\alpha^{n-1}(x_j)]_C)\psi([x_k,x_l]_C,\alpha(x_1),\dots,\widehat{\alpha(x_k)},\\
\nonumber&&\dots,\widehat{\alpha(x_l)},\dots,\widehat{\alpha(x_i)},\dots,\widehat{\alpha(x_j)},\dots,\alpha(x_{n+2}))\\
\label{eq-124}&&+\sum_{1\leq k<i<l<j\leq n+1}(-1)^{i+j}(-1)^{k+l}\rho([\alpha^{n-1}(x_i),\alpha^{n-1}(x_j)]_C)\psi([x_k,x_l]_C,\alpha(x_1),\dots,\widehat{\alpha(x_k)},\\
\nonumber&&\dots,\widehat{\alpha(x_i)},\dots,\widehat{\alpha(x_l)},\dots,\widehat{\alpha(x_j)},\dots,\alpha(x_{n+2}))\\
\label{eq-125}&&+\sum_{1\leq k<i<j<l\leq n+1}(-1)^{i+j}(-1)^{k+l-1}\rho([\alpha^{n-1}(x_i),\alpha^{n-1}(x_j)]_C)\psi([x_k,x_l]_C,\alpha(x_1),\dots,\widehat{\alpha(x_k)},\\
\nonumber&&\dots,\widehat{\alpha(x_i)},\dots,\widehat{\alpha(x_j)},\dots,\widehat{\alpha(x_l)},\dots,\alpha(x_{n+2}))\\
\label{eq-126}&&+\sum_{1\leq i<k<l<j\leq n+1}(-1)^{i+j}(-1)^{k+l-1}\rho([\alpha^{n-1}(x_i),\alpha^{n-1}(x_j)]_C)\psi([x_k,x_l]_C,\alpha(x_1),\dots,\widehat{\alpha(x_i)},\\
\nonumber&&\dots,\widehat{\alpha(x_k)},\dots,\widehat{\alpha(x_l)},\dots,\widehat{\alpha(x_j)},\dots,\alpha(x_{n+2}))\\
\label{eq-127}&&+\sum_{1\leq i<k<j<l\leq n+1}(-1)^{i+j}(-1)^{k+l}\rho([\alpha^{n-1}(x_i),\alpha^{n-1}(x_j)]_C)\psi([x_k,x_l]_C,\alpha(x_1),\dots,\widehat{\alpha(x_i)},\\
\nonumber&&\dots,\widehat{\alpha(x_k)},\dots,\widehat{\alpha(x_j)},\dots,\widehat{\alpha(x_l)},\dots,\alpha(x_{n+2}))\\
\label{eq-128}&&+\sum_{1\leq i<j<k<l\leq n+1}(-1)^{i+j}(-1)^{k+l-1}\rho([\alpha^{n-1}(x_i),\alpha^{n-1}(x_j)]_C)\psi([x_k,x_l]_C,\alpha(x_1),\dots,\widehat{\alpha(x_i)},\\
\nonumber&&\dots,\widehat{\alpha(x_j)},\dots,\widehat{\alpha(x_k)},\dots,\widehat{\alpha(x_l)},\dots,\alpha(x_{n+2}))\\
\label{eq-129}&&+\sum_{1\leq k<i<j\leq n+1}(-1)^{i+j}(-1)^k\mu(\alpha^{n-1}(x_i)\cdot\alpha^{n-1}(x_{n+2}))\rho(\alpha^{n-1}(x_k))\psi(x_1,\dots,\hat{x_k},\dots,\hat{x_i},
\dots,\hat{x_j}\dots,x_{n+1},x_j)\\
\label{eq-130}&&+\sum_{1\leq i<k<j\leq n+1}(-1)^{i+j}(-1)^{k+1}\mu(\alpha^{n-1}(x_i)\cdot\alpha^{n-1}(x_{n+2}))\rho(\alpha^{n-1}(x_k))\psi(x_1,\dots,\hat{x_i},\dots,\hat{x_k},
\dots,\hat{x_j}\dots,x_{n+1},x_j)\\
\label{eq-131}&&+\sum_{1\leq i<j<k\leq n+1}(-1)^{i+j}(-1)^k\mu(\alpha^{n-1}(x_i)\cdot\alpha^{n-1}(x_{n+2}))\rho(\alpha^{n-1}(x_k))\psi(x_1,\dots,\hat{x_i},\dots,\hat{x_j},
\dots,\hat{x_k}\dots,x_{n+1},x_j)\\
\label{eq-132}&&+\sum_{1\leq k<i<j\leq n+1}(-1)^{i+j}(-1)^k\mu(\alpha^{n-1}(x_i)\cdot\alpha^{n-1}(x_{n+2}))\mu(\alpha^{n-1}(x_j))\psi(x_1,\dots,\hat{x_k},\dots,\hat{x_i},
\dots,\hat{x_j}\dots,x_{n+1},x_k)\\
\label{eq-133}&&+\sum_{1\leq i<k<j\leq n+1}(-1)^{i+j}(-1)^{k+1}\mu(\alpha^{n-1}(x_i)\cdot\alpha^{n-1}(x_{n+2}))\mu(\alpha^{n-1}(x_j))\psi(x_1,\dots,\hat{x_i},\dots,\hat{x_k},
\dots,\hat{x_j}\dots,x_{n+1},x_k)\\
\label{eq-134}&&+\sum_{1\leq i<j<k\leq n+1}(-1)^{i+j}(-1)^k\mu(\alpha^{n-1}(x_i)\cdot\alpha^{n-1}(x_{n+2}))\mu(\alpha^{n-1}(x_j))\psi(x_1,\dots,\hat{x_i},\dots,\hat{x_j},
\dots,\hat{x_k}\dots,x_{n+1},x_k)\\
\label{eq-135}&&-\sum_{1\leq k<i<j\leq n+1}(-1)^{i+j}(-1)^k\mu(\alpha^{n-1}(x_i)\cdot\alpha^{n-1}(x_{n+2}))\psi(\alpha(x_1),\dots,\widehat{\alpha(x_k)},\dots,\widehat{\alpha(x_i)},\\
\nonumber&&\dots,\widehat{\alpha(x_j)},\dots,\alpha(x_{n+1}),x_k\cdot x_j)\\
\label{eq-136}&&-\sum_{1\leq i<k<j\leq n+1}(-1)^{i+j}(-1)^{k+1}\mu(\alpha^{n-1}(x_i)\cdot\alpha^{n-1}(x_{n+2}))\psi(\alpha(x_1),\dots,\widehat{\alpha(x_i)},\dots,\widehat{\alpha(x_k)},\\
\nonumber&&\dots,\widehat{\alpha(x_j)},\dots,\alpha(x_{n+1}),x_k\cdot x_j)\\
\label{eq-137}&&-\sum_{1\leq i<j<k\leq n+1}(-1)^{i+j}(-1)^k\mu(\alpha^{n-1}(x_i)\cdot\alpha^{n-1}(x_{n+2}))\psi(\alpha(x_1),\dots,\widehat{\alpha(x_i)},\dots,\widehat{\alpha(x_j)},\\
\nonumber&&\dots,\widehat{\alpha(x_k)},\dots,\alpha(x_{n+1}),x_k\cdot x_j)\\
\label{eq-138}&&+\sum_{1\leq k<l<i<j\leq n+1}(-1)^{i+j}(-1)^{k+l-1}\mu(\alpha^{n-1}(x_i)\cdot\alpha^{n-1}(x_{n+2}))\psi([x_k,x_l]_C,\alpha(x_1),\dots,\widehat{\alpha(x_k)},\\
\nonumber&&\dots,\widehat{\alpha(x_l)},\dots,\widehat{\alpha(x_i)},\dots,\widehat{\alpha(x_j)},\dots,\alpha(x_{n+1}),\alpha(x_j))\\
\label{eq-139}&&+\sum_{1\leq k<i<l<j\leq n+1}(-1)^{i+j}(-1)^{k+l}\mu(\alpha^{n-1}(x_i)\cdot\alpha^{n-1}(x_{n+2}))\psi([x_k,x_l]_C,\alpha(x_1),\dots,\widehat{\alpha(x_k)},\\
\nonumber&&\dots,\widehat{\alpha(x_i)},\dots,\widehat{\alpha(x_l)},\dots,\widehat{\alpha(x_j)},\dots,\alpha(x_{n+1}),\alpha(x_j))\\
\label{eq-140}&&+\sum_{1\leq k<i<j<l\leq n+1}(-1)^{i+j}(-1)^{k+l-1}\mu(\alpha^{n-1}(x_i)\cdot\alpha^{n-1}(x_{n+2}))\psi([x_k,x_l]_C,\alpha(x_1),\dots,\widehat{\alpha(x_k)},\\
\nonumber&&\dots,\widehat{\alpha(x_i)},\dots,\widehat{\alpha(x_j)},\dots,\widehat{\alpha(x_l)},\dots,\alpha(x_{n+1}),\alpha(x_j))\\
\label{eq-141}&&+\sum_{1\leq i<k<l<j\leq n+1}(-1)^{i+j}(-1)^{k+l-1}\mu(\alpha^{n-1}(x_i)\cdot\alpha^{n-1}(x_{n+2}))\psi([x_k,x_l]_C,\alpha(x_1),\dots,\widehat{\alpha(x_i)},\\
\nonumber&&\dots,\widehat{\alpha(x_k)},\dots,\widehat{\alpha(x_l)},\dots,\widehat{\alpha(x_j)},\dots,\alpha(x_{n+1}),\alpha(x_j))\\
\label{eq-142}&&+\sum_{1\leq i<k<j<l\leq n+1}(-1)^{i+j}(-1)^{k+l}\mu(\alpha^{n-1}(x_i)\cdot\alpha^{n-1}(x_{n+2}))\psi([x_k,x_l]_C,\alpha(x_1),\dots,\widehat{\alpha(x_i)},\\
\nonumber&&\dots,\widehat{\alpha(x_k)},\dots,\widehat{\alpha(x_j)},\dots,\widehat{\alpha(x_l)},\dots,\alpha(x_{n+1}),\alpha(x_j))\\
\label{eq-143}&&+\sum_{1\leq i<j<k<l\leq n+1}(-1)^{i+j}(-1)^{k+l-1}\mu(\alpha^{n-1}(x_i)\cdot\alpha^{n-1}(x_{n+2}))\psi([x_k,x_l]_C,\alpha(x_1),\dots,\widehat{\alpha(x_i)},\\
\nonumber&&\dots,\widehat{\alpha(x_j)},\dots,\widehat{\alpha(x_k)},\dots,\widehat{\alpha(x_l)},\dots,\alpha(x_{n+1}),\alpha(x_j))\\
\label{eq-144}&&-\sum_{1\leq k<i<j\leq n+1}(-1)^{i+j}(-1)^k\mu(\alpha^{n-1}(x_j)\cdot\alpha^{n-1}(x_{n+2}))\rho(\alpha^{n-1}(x_k))\psi(x_1,\dots,\hat{x_k},\dots,\hat{x_i},
\dots,\hat{x_j}\dots,x_{n+1},x_i)\\
\label{eq-145}&&-\sum_{1\leq i<k<j\leq n+1}(-1)^{i+j}(-1)^{k+1}\mu(\alpha^{n-1}(x_j)\cdot\alpha^{n-1}(x_{n+2}))\rho(\alpha^{n-1}(x_k))\psi(x_1,\dots,\hat{x_i},\dots,\hat{x_k},
\dots,\hat{x_j}\dots,x_{n+1},x_i)\\
\label{eq-146}&&-\sum_{1\leq i<j<k\leq n+1}(-1)^{i+j}(-1)^k\mu(\alpha^{n-1}(x_j)\cdot\alpha^{n-1}(x_{n+2}))\rho(\alpha^{n-1}(x_k))\psi(x_1,\dots,\hat{x_i},\dots,\hat{x_j},
\dots,\hat{x_k}\dots,x_{n+1},x_i)\\
\label{eq-147}&&-\sum_{1\leq k<i<j\leq n+1}(-1)^{i+j}(-1)^k\mu(\alpha^{n-1}(x_j)\cdot\alpha^{n-1}(x_{n+2}))\mu(\alpha^{n-1}(x_i))\psi(x_1,\dots,\hat{x_k},\dots,\hat{x_i},
\dots,\hat{x_j}\dots,x_{n+1},x_k)\\
\label{eq-148}&&-\sum_{1\leq i<k<j\leq n+1}(-1)^{i+j}(-1)^{k+1}\mu(\alpha^{n-1}(x_j)\cdot\alpha^{n-1}(x_{n+2}))\mu(\alpha^{n-1}(x_i))\psi(x_1,\dots,\hat{x_i},\dots,\hat{x_k},
\dots,\hat{x_j}\dots,x_{n+1},x_k)\\
\label{eq-149}&&-\sum_{1\leq i<j<k\leq n+1}(-1)^{i+j}(-1)^k\mu(\alpha^{n-1}(x_j)\cdot\alpha^{n-1}(x_{n+2}))\mu(\alpha^{n-1}(x_i))\psi(x_1,\dots,\hat{x_i},\dots,\hat{x_j},
\dots,\hat{x_k}\dots,x_{n+1},x_k)\\
\label{eq-150}&&+\sum_{1\leq k<i<j\leq n+1}(-1)^{i+j}(-1)^k\mu(\alpha^{n-1}(x_j)\cdot\alpha^{n-1}(x_{n+2}))\psi(\alpha(x_1),\dots,\widehat{\alpha(x_k)},\dots,\widehat{\alpha(x_i)},\\
\nonumber&&\dots,\widehat{\alpha(x_j)},\dots,\alpha(x_{n+1}),x_k\cdot x_i)\\
\label{eq-151}&&+\sum_{1\leq i<k<j\leq n+1}(-1)^{i+j}(-1)^{k+1}\mu(\alpha^{n-1}(x_j)\cdot\alpha^{n-1}(x_{n+2}))\psi(\alpha(x_1),\dots,\widehat{\alpha(x_i)},\dots,\widehat{\alpha(x_k)},\\
\nonumber&&\dots,\widehat{\alpha(x_j)},\dots,\alpha(x_{n+1}),x_k\cdot x_i)\\
\label{eq-152}&&+\sum_{1\leq i<j<k\leq n+1}(-1)^{i+j}(-1)^k\mu(\alpha^{n-1}(x_j)\cdot\alpha^{n-1}(x_{n+2}))\psi(\alpha(x_1),\dots,\widehat{\alpha(x_i)},\dots,\widehat{\alpha(x_j)},\\
\nonumber&&\dots,\widehat{\alpha(x_k)},\dots,\alpha(x_{n+1}),x_k\cdot x_i)\\
\label{eq-153}&&-\sum_{1\leq k<l<i<j\leq n+1}(-1)^{i+j}(-1)^{k+l-1}\mu(\alpha^{n-1}(x_j)\cdot\alpha^{n-1}(x_{n+2}))\psi([x_k,x_l]_C,\alpha(x_1),\dots,\widehat{\alpha(x_k)},\\
\nonumber&&\dots,\widehat{\alpha(x_l)},\dots,\widehat{\alpha(x_i)},\dots,\widehat{\alpha(x_j)},\dots,\alpha(x_{n+1}),\alpha(x_i))\\
\label{eq-154}&&-\sum_{1\leq k<i<l<j\leq n+1}(-1)^{i+j}(-1)^{k+l}\mu(\alpha^{n-1}(x_j)\cdot\alpha^{n-1}(x_{n+2}))\psi([x_k,x_l]_C,\alpha(x_1),\dots,\widehat{\alpha(x_k)},\\
\nonumber&&\dots,\widehat{\alpha(x_i)},\dots,\widehat{\alpha(x_l)},\dots,\widehat{\alpha(x_j)},\dots,\alpha(x_{n+1}),\alpha(x_i))\\
\label{eq-155}&&-\sum_{1\leq k<i<j<l\leq n+1}(-1)^{i+j}(-1)^{k+l-1}\mu(\alpha^{n-1}(x_j)\cdot\alpha^{n-1}(x_{n+2}))\psi([x_k,x_l]_C,\alpha(x_1),\dots,\widehat{\alpha(x_k)},\\
\nonumber&&\dots,\widehat{\alpha(x_i)},\dots,\widehat{\alpha(x_j)},\dots,\widehat{\alpha(x_l)},\dots,\alpha(x_{n+1}),\alpha(x_i))\\
\label{eq-156}&&-\sum_{1\leq i<k<l<j\leq n+1}(-1)^{i+j}(-1)^{k+l-1}\mu(\alpha^{n-1}(x_j)\cdot\alpha^{n-1}(x_{n+2}))\psi([x_k,x_l]_C,\alpha(x_1),\dots,\widehat{\alpha(x_i)},\\
\nonumber&&\dots,\widehat{\alpha(x_k)},\dots,\widehat{\alpha(x_l)},\dots,\widehat{\alpha(x_j)},\dots,\alpha(x_{n+1}),\alpha(x_i))\\
\label{eq-157}&&-\sum_{1\leq i<k<j<l\leq n+1}(-1)^{i+j}(-1)^{k+l}\mu(\alpha^{n-1}(x_j)\cdot\alpha^{n-1}(x_{n+2}))\psi([x_k,x_l]_C,\alpha(x_1),\dots,\widehat{\alpha(x_i)},\\
\nonumber&&\dots,\widehat{\alpha(x_k)},\dots,\widehat{\alpha(x_j)},\dots,\widehat{\alpha(x_l)},\dots,\alpha(x_{n+1}),\alpha(x_i))\\
\label{eq-158}&&-\sum_{1\leq i<j<k<l\leq n+1}(-1)^{i+j}(-1)^{k+l-1}\mu(\alpha^{n-1}(x_j)\cdot\alpha^{n-1}(x_{n+2}))\psi([x_k,x_l]_C,\alpha(x_1),\dots,\widehat{\alpha(x_i)},\\
\nonumber&&\dots,\widehat{\alpha(x_j)},\dots,\widehat{\alpha(x_k)},\dots,\widehat{\alpha(x_l)},\dots,\alpha(x_{n+1}),\alpha(x_i)).
\end{eqnarray}
}
By the definition of the sub-adjacent Hom-Lie algebra, the sum of \eqref{eq-93}, \eqref{eq-94} and \eqref{eq-95} is zero.  By the definition of Hom-pre-Lie algebras, the sum of \eqref{eq-81}, \eqref{eq-86} and \eqref{eq-109} is zero, the sum of \eqref{eq-82}, \eqref{eq-83} and \eqref{eq-110} is zero, the sum of \eqref{eq-84}, \eqref{eq-85} and \eqref{eq-108} is zero. Obviously, the sum of \eqref{eq-87} and \eqref{eq-128} is zero, the sum of \eqref{eq-88} and \eqref{eq-127} is zero, the sum of \eqref{eq-89} and \eqref{eq-126} is zero, the sum of \eqref{eq-90} and \eqref{eq-125} is zero, the sum of \eqref{eq-91} and \eqref{eq-124} is zero, the sum of \eqref{eq-92} and \eqref{eq-123} is zero, the sum of \eqref{eq-78} and \eqref{eq-122} is zero, the sum of \eqref{eq-79} and \eqref{eq-121} is zero, the sum of \eqref{eq-80} and \eqref{eq-120} is zero, the sum of \eqref{eq-96} and \eqref{eq-143} is zero, the sum of \eqref{eq-97} and \eqref{eq-142} is zero, the sum of \eqref{eq-98} and \eqref{eq-141} is zero, the sum of \eqref{eq-99} and \eqref{eq-140} is zero, the sum of \eqref{eq-100} and \eqref{eq-139} is zero, the sum of \eqref{eq-101} and \eqref{eq-138} is zero, the sum of \eqref{eq-102} and \eqref{eq-158} is zero, the sum of \eqref{eq-103} and \eqref{eq-157} is zero, the sum of \eqref{eq-104} and \eqref{eq-156} is zero, the sum of \eqref{eq-105} and \eqref{eq-155} is zero, the sum of \eqref{eq-106} and \eqref{eq-154} is zero, the sum of \eqref{eq-107} and \eqref{eq-153} is zero. By the definition of the sub-adjacent Lie bracket, the sum of \eqref{eq-111}, \eqref{eq-136} and \eqref{eq-137} is zero, the sum of \eqref{eq-112}, \eqref{eq-135} and \eqref{eq-152} is zero, the sum of \eqref{eq-113}, \eqref{eq-150} and \eqref{eq-151} is zero. Since $(V,\beta,\rho)$ is a representation of the sub-adjacent Hom-Lie algebra $A^C$, the sum of \eqref{eq-60}, \eqref{eq-61}, \eqref{eq-62}, \eqref{eq-114}, \eqref{eq-115} and \eqref{eq-116} is zero. Since $(V,\beta,\rho,\mu)$ is a representation of the Hom-pre-Lie algebra $(A,\cdot,\alpha)$, the sum of \eqref{eq-64}, \eqref{eq-65}, \eqref{eq-69}, \eqref{eq-73}, \eqref{eq-74}, \eqref{eq-119}, \eqref{eq-129}, \eqref{eq-130}, \eqref{eq-134} and \eqref{eq-149} is zero, the sum of \eqref{eq-63}, \eqref{eq-68}, \eqref{eq-70}, \eqref{eq-72}, \eqref{eq-77}, \eqref{eq-118}, \eqref{eq-131}, \eqref{eq-133}, \eqref{eq-144} and \eqref{eq-148} is zero, the sum of \eqref{eq-66}, \eqref{eq-67}, \eqref{eq-71}, \eqref{eq-75}, \eqref{eq-76}, \eqref{eq-117}, \eqref{eq-132}, \eqref{eq-145}, \eqref{eq-146} and \eqref{eq-147} is zero. Thus, we have $\partial_{\omega\omega}\circ \partial_{\alpha\omega}+\partial_{\alpha\omega}\circ \partial_{\alpha\alpha}=0$.

Similarly, we have $\partial_{\omega\alpha}\circ \partial_{\omega\omega}+\partial_{\alpha\alpha}\circ \partial_{\omega\alpha}=0$ and $\partial_{\omega\alpha}\circ \partial_{\alpha\omega}+\partial_{\alpha\alpha}\circ \partial_{\alpha\alpha}=0$. This finishes the proof.

 \end{document}